%% file: main_camera_ready.tex
\documentclass[nohyperref]{article} 
\usepackage[accepted]{icml2023}

\input{math_commands.tex}

\usepackage{hyperref}
\usepackage{url}
\usepackage{microtype}
\usepackage{booktabs}
\usepackage{nicefrac}       
\usepackage{microtype}      
\usepackage{xcolor}         
\usepackage{graphicx}
\usepackage{enumerate,enumitem}
\usepackage{multirow}
\usepackage{caption}
\usepackage{subcaption}
\usepackage{float}
\usepackage{todonotes}
\usepackage{algorithm,algorithmic}
\usepackage[section]{placeins}
\usepackage{wrapfig}
\usepackage[toc,page,header]{appendix}
\usepackage{minitoc}

\setcitestyle{square,numbers}

\newtheorem{theorem}{Theorem}[section]

\newtheorem{definition}{Definition}

\newtheorem{remark}[theorem]{Remark}

\newcommand{\emptycomment}[1]{}

\icmltitlerunning{On Heterogeneous Treatment Effects in 
Heterogeneous Causal Graphs}

\begin{document}

\doparttoc 
\faketableofcontents 

\linepenalty=1000

\setlength{\abovedisplayskip}{3pt}
\setlength{\belowdisplayskip}{3pt}
\setlength{\abovedisplayshortskip}{3pt}
\setlength{\belowdisplayshortskip}{3pt}

\twocolumn[
\icmltitle{\bf On Heterogeneous Treatment Effects in 
Heterogeneous Causal Graphs}



\icmlsetsymbol{equal}{*}

\begin{icmlauthorlist}
\icmlauthor{Richard Watson}{equal,ncsu}
\icmlauthor{Hengrui Cai}{equal,uci}
\icmlauthor{Xinming An}{unc}
\icmlauthor{Samuel Mclean}{unc}
\icmlauthor{Rui Song}{ncsu}
\end{icmlauthorlist}

\icmlaffiliation{ncsu}{Department of Statistics, North Carolina State University, 
USA}
\icmlaffiliation{uci}{Department of Statistics, University of California Irvine, 
USA}
\icmlaffiliation{unc}{Department of Anesthesiology, University of North Carolina Chapel Hill, 
USA}

\icmlcorrespondingauthor{Rui Song}{songray@gmail.com}

\icmlkeywords{Causal Inference, Causal Graphs, Causal Discovery, Heterogeneous Casual Effects, Machine Learning, ICML}

\vskip 0.3in
]



\printAffiliationsAndNotice{\icmlEqualContribution} 

\begin{abstract}
Heterogeneity and comorbidity are two interwoven challenges associated with various healthcare problems that greatly hampered research on developing effective treatment and understanding of the underlying neurobiological mechanism. Very few studies have been conducted to investigate \textit{heterogeneous causal effects} (HCEs) in graphical contexts due to the lack of statistical methods. To characterize this heterogeneity, we first conceptualize \textit{heterogeneous causal graphs} (HCGs) by generalizing the causal graphical model with confounder-based interactions and multiple mediators. Such confounders with an interaction with the treatment are known as \textit{moderators}. This allows us to flexibly produce HCGs given different moderators and explicitly characterize HCEs from the treatment or potential mediators on the outcome. We establish the theoretical forms of HCEs and derive their properties at the individual level in both linear and nonlinear models. An interactive structural learning is developed to estimate the complex HCGs and HCEs with confidence intervals provided. Our method is empirically justified by extensive simulations and its practical usefulness is illustrated by exploring causality among psychiatric disorders for trauma survivors. 
Code implementing the proposed algorithm is open-source and publicly available at: \url{https://github.com/richard-watson/ISL}.
\end{abstract}
\section{Introduction}

During the last several decades, little progress has been made in the understanding of the underlying neurobiological mechanism and in developing effective treatments for psychiatric disorders due to several unique challenges \citep{mclean2020aurora}. Accumulating evidence from existing research suggests that \textit{Heterogeneity} is a common issue for psychiatric disorders based on traditional classification and diagnoses. Most of the existing studies that address heterogeneity based on symptom severity levels \citep{feczko2019heterogeneity} are less likely to provide new insights about underlying neurobiological mechanisms. On the other hand, studying heterogeneous causal relationships of different psychiatric disorders has the potential to improve understanding, risk prediction, and treatment selection. However, to date, no such studies have been conducted due to the lack of appropriate analysis techniques. To fill in this gap and advance research for psychiatric disorders, we developed an innovative analysis approach to investigate and test heterogeneous causal relationships and applied it to study psychiatric disorders after trauma exposure using data from a large cohort study.

\textit{Causal discovery}, the task of discovering the causal relations between variables in a dataset, has interesting and important applications in many areas, such as epidemiology \citep{hernan2004definition}, medicine \citep{hernan2000marginal},  economics \citep{panizza2014public}, etc. Under a general causal graph, an event or a treatment may have a direct effect on the outcome of interest as well as an indirect effect regulated by a set of \textit{mediators} (which are variables that are affected by the treatment and in turn affect the outcome as $M$ does in Figure \ref{fig1}) \citep[see an overview in][]{pearl2009causal}. Yet, due to the heterogeneity of individuals  in response to different events/treatments, there may not exist a uniform causal graph for everyone. This implies the existence of heterogeneous causal graphs under different \textit{moderators} which are pre-treatment confounders that also have an interaction term with the treatment as $\bm{X}$ does in Figure \ref{fig1}). Learning such heterogeneity is one of the most crucial obstacles that continue to hamper advances in many fields such as psychiatric disorders \citep{marquand2016beyond}. With the launch of the Advancing Understanding of RecOvery afteR traumA (AURORA) study \citep{mclean2020aurora}, where thousands of trauma survivors were recruited after traumatic experiences and followed to collect a broad range of bio-behavioral data, discovering  heterogeneous causal patterns thus becomes a timely issue.

\begin{figure*}[!htb]
\centering
\begin{subfigure}{.31\linewidth}
  \centering
  \includegraphics[width=\linewidth]{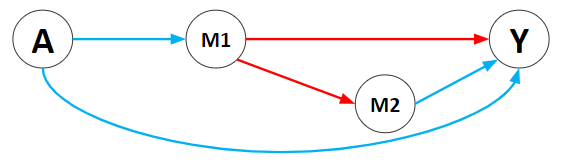}
\end{subfigure}%
\begin{subfigure}{.35\linewidth}
  \centering
  \includegraphics[width=\linewidth]{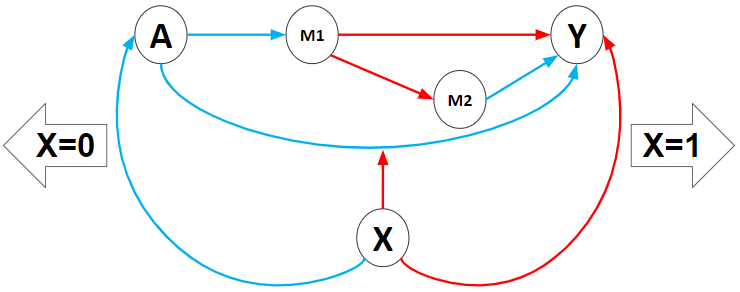}
\end{subfigure}%
\begin{subfigure}{.31\linewidth}
  \centering
  \includegraphics[width=\linewidth]{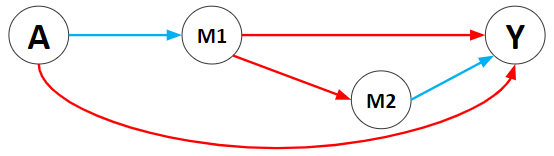}
\end{subfigure}%
\vspace{-0.2cm}
\caption{The middle panel is the whole causal graph, and the left and right graphs are heterogeneous causal graphs for different values of $\bm{X}$. Here $\bm{X}$ confounds the relationship between $A$ (the treatment or event) and $Y$ (the outcome of interest), and modifies the direct effect of $A$ on $Y$ dependent on the value of $\bm{X}$. $M_1$ and $M_2$ are two different mediators that mediate the indirect effect of $A$ onto $Y$. The red and blue arrows represent whether the causal relation is positive or negative, respectively.}
    \label{fig1}
\end{figure*}

Despite the popularity of causal discovery methods \citep[e.g.,][]{spirtes2000constructing,shimizu2006linear,kalisch2007estimating,buhlmann2014cam,ramsey2017million,zheng2018dags,yu2019dag,zhu2019causal}, none of these works considers heterogeneous causal graphs to account for the heterogeneity among different subjects. On the other hand, all existing works of learning heterogeneity in causal inference literature \citep[e.g.,][]{wager2018estimation,kunzel2019metalearners,nie2021quasi} either assume a known homogeneous causal graph or completely ignore the graphical structure, and thus fail to capture the heterogeneous causal pathways. In this paper, we focus on learning heterogeneous causal graphs and effects with multiple (possibly sequentially ordered) mediators. Our \textbf{contributions} are three-fold. 

\hspace{0.2cm} $\bullet$\hspace{0.2cm}  To our knowledge, this is the first work that considers heterogeneity in terms of causal graphs. We first conceptualize \textit{heterogeneous causal graphs} (HCGs), by incorporating the information of moderators and their interactions into the causal graphical model. This introduces several unique challenges in presenting the interacted nodes in the graphical model and handling such interactions in a consistent self-generated system. We address these difficulties by proposing a novel structural equation model with interactions to flexibly produce HCGs from a hybrid causal graph.

\hspace{0.2cm} $\bullet$ \hspace{0.2cm} We propose \textit{heterogeneous causal effects} (HCEs) to systematically quantify the impact of treatments and mediators on the outcome of interest given different values of moderators. Based on the proposed graphical model with interactions, we establish the explicit forms of the defined HCEs and derive their theoretical properties at the individual level, in either \textbf{linear} or \textbf{non-linear} structural equation model. 

\hspace{0.2cm} $\bullet$\hspace{0.2cm}  We develop a comprehensive procedure for extracting heterogeneous causal reasoning. Specifically, we propose an interactive structural learning algorithm to learn the complex HCGs, estimate HCEs, and compute bootstrap confidence intervals for these estimates via a debiasing process. Our method is general enough to address most causal structures.
\subsection{Related Works}
\textbf{Conditional average treatment effect.} A vast number of approaches have been proposed for estimating the heterogeneity via the conditional average treatment effect (CATE) that quantifies the effect size of treatment on the outcome of interest given different confounders \citep[see recent advances in][]{athey2015machine,shalit2017estimating,wager2018estimation,kunzel2019metalearners,nie2021quasi,farrell2021deep}. Here, confounders used in CATE usually play the role of the moderators \citep{kraemer2002mediators} which modify the impact of the treatment on the outcome as a presence of confounder-treatment interaction in modeling the outcome. All these works neglect the underlying causal structure among potential mediators 
in regulating the treatment effects and thus lead to a less interpretable causal mechanism.

\textbf{Traditional mediation analysis.} There is an abundance of literature for traditional mediation analysis \citep[see a review in][and the reference therein]{van2016med} quantifying treatment effects in the presence of multiple mediators. Due to multiple causally dependent mediators, \citet{robin2003semantics} found that causal effects are not identifiable unless no treatment-mediator interaction is imposed, with more discussions on the identification of causal effects later on \citep{imai2013identification,tchetgen2014mediation,van2019med}. 
Yet, all of these works primarily assumed a known causal graph with  only a few mediators included for analysis. In addition, the moderators in traditional mediation analysis \citep{muller2005modmed} interact with the treatment but are  independent of the treatment. In this work, we do not enforce such an independence requirement.


\textbf{Causal structural learning.} Plentiful causal structural learning (CSL) methods have been proposed to learn the unknown causal
structure within a class of directed acyclic graphs from observed data. The large literature can be categorized into three types. The testing-based methods \citep[e.g.,][for the well-known PC algorithm]{spirtes2000constructing} rely on the conditional independence tests to find the causal skeleton and edge orientations under the linear structural equation model (SEM). Based on additional and proper assumptions on noises and models, the functional-based methods handle both linear SEM \citep[e.g.,][]{shimizu2006linear} and non-linear SEM \citep[e.g.,][]{buhlmann2014cam}. Recently, the score-based methods formulate the CSL problem into optimization by certain score functions, for both linear SEM \citep[e.g.,][]{ramsey2017million,zheng2018dags} and non-linear SEM \citep[e.g.,][]{yu2019dag,zhu2019causal,zheng2020learning,rolland2022score}. Yet, all these works neglect causal contexts (i.e., treatments, mediators, and outcomes) and thus cannot develop HCGs.

\textbf{Causal mediation analysis with CSL.} To visualize causes and counterfactuals, \citet{pearl2009causal} proposed to use the causal graphical model and the `do-operator' to quantify the causal effects. A number of follow-up works \citep[e.g.,][]{maathuis2009estimating,nandy2017estimating,chakrabortty2018inference} have been developed recently to estimate direct and indirect causal effects that are regulated by mediators in the linear SEM. These studies relied on the PC algorithm \citep{spirtes2000constructing} which requires strong assumptions of graph sparsity and noise normality due to computational limits. To overcome these difficulties,  \citet{cai2020anoce} proposed to leverage score-based CSL methods \citep[e.g.,][]{ramsey2017million,zheng2018dags,yu2019dag,zhu2019causal} with background causal knowledge to estimate mediation effects. These works assumed no moderator in their linear SEMs, such that the causal graph or effect learned is on the population level, and thus cannot access the heterogeneity.

\section{Framework}\label{sec:assump}
\textbf{Causal graph terminology.} A directed acyclic graph (DAG) is defined as a directed graph that does not contain directed cycles. Specifically, let a DAG $\mathbb{G} =(\bm{Z},\bm{E})$ characterize the causal relationship among a set of random variables $\bm{Z}$ and an edge set $\bm{E}$. Here, $\bm{Z}=\{Z_1,\dots,Z_w\}$ represents a random vector of $w=|\bm{Z}|$ nodes and an edge $Z_i\rightarrow Z_j \in \bm{E}$ means that $Z_i$ is a direct cause of $Z_j$. A variable node $Z_i$ is said to be a parent of $Z_j$ if there is a directed edge from $Z_i$ to $Z_j$. The set of all parents of node $Z_j$ in $\mathbb{G}$ is denoted as $PA_{ \mathbb{G}}({Z_j})$. Let $\bm{B}=\{b_{i,j}\}_{1\leq i\leq w,1\leq j\leq w}$ be a $w\times w$ matrix, where $b_{i,j}$ is the weight of the edge $Z_i\rightarrow Z_j \in \bm{E}$, and $b_{i,j}=0$ otherwise. Then, we say that $\mathbb{G}=(\bm{Z},\bm{B})$ is a weighted DAG with the variable/node set $\bm{Z}$ and the weighted adjacency matrix $\bm{B}$.

\textbf{Notations and assumptions.} Let $\bm{X}=[X_1,X_2,\dots,X_p]^\top$ be a $p$-dimensional vector of baseline information prior to the event or the treatment $A$,  $\bm{M}=[M_1,M_2,\dots,M_s]^\top$ be a $s$-dimensional vector of post-treatment possible mediators, and $Y$ be the final outcome of interest. We define $\bm{X}$ generally as a vector of possible moderators that precede the treatment with no independence restriction. There may be no interaction between $X_1$ and $A$, in which case it is not a moderator, but a confounder. Further $X_2$ may only affect $Y$, in which case it is a covariate. We allow $\bm{X}$ to be any combination of covariates, confounders, and moderators. As commonly imposed in CSL works \citep[e.g.,][]{spirtes2000causation,peters2014causal}, we assume the Markov and faithfulness 
assumptions (see explicit definitions and further discussion in Appendix \ref{asec:assump}). Let $Y^*\equiv Y^*(A=a, \bm{M}=\bm{m})$ when $A=a$ and $\bm{M}=\bm{m}$ as the potential outcome and $M_i^* \equiv M_i^*(A=a,M_j=m_j)$ when $A=a$ and $M_j=m_j$ as the potential mediator. For any mediator $M_i$ and any of its parent $M_j\in PA_{\mathbb{G}}(M_i)$\textbackslash$\{\bm{X},A\}$, we assume no unmeasured confounders for:\\
    \textbf{(A1)}. the treatment on the outcome and mediators:
   \begin{align*}
        Y^* \independent A|\bm{X};   \quad  \quad
        M_i^*  \independent A|\bm{X};
   \end{align*}
   \textbf{(A2)}. the mediators on the outcome:
   \begin{align*}
       \quad  Y^* \independent M_j|\{A,\bm{X}\}; \quad  \quad Y^* \independent M_i|\{A,\bm{X},M_j\};
    \end{align*}
   \textbf{(A3)}. the potential mediators:
   \begin{align*}
      & Y^* \independent M_j^*|\bm{X}; \quad   M_i^* \independent M_j^*|\bm{X};\quad   Y^* \independent M_i^*|\{M_j^*,\bm{X}\}. \end{align*}
Assumptions (A1-A3) are concerned with the completeness of the data as common in causal inference \citep[e.g.,][]{pearl2009causal,nandy2017estimating}, where we assume that all variables causally related to any variable in the data are included in data. Also note that we \textbf{make no assumptions on the structure of the mediators} which allows our method to take into account the fact that not all potential mediators are true mediators.  Finally, the causal graph, in general, is only identifiable to its Markov Equivalence Class (MEC), however if, for example, the data is assumed to be Gaussian with equal variance, it can be identified from its MEC \citep{peters2014identifiability} (see details in Appendix \ref{sec:identif}).
\section{Heterogeneous Causal Graphs}  
We conceptualize the heterogeneous causal graph (HCG), by generalizing
the causal graphical model with moderation.
\subsection{
Structural Equation Model 
with  Interactions}\label{sec:lsem_intera}
 Considering interactions among variables  introduces several unique challenges in the causal graphical model, for example, presenting the interacted nodes in the graph and modeling  heterogeneous interactions in a consistent self-generated system. To address these issues, we first generalize the linear SEM for the causal graphical model with interactions, which yields good interpretation to easily link the parameters with causal effects to be specified (see the functional SEM version of our method to handle complex relationships in Section \ref{ext}). 
 For simplicity of exposition, we use the product of $\bm{X}$ and $A$, $\bm{X}A$, as the interaction term (see discussions on the importance of such an interaction term in Remark \ref{rem:whyinter}). $\bm{D}=[\bm{X},A,\bm{X}A,\bm{M},Y]^\top$ is a $(2p+s+2)$-dimensional data vector for an individual. Suppose there exists a weighted adjacency matrix  $\bm{B}$  such that $\bm{D}$ and $\bm{B}$ together form a DAG, denoted as $\mathbb{G}=(\bm{D},\bm{B})$. This DAG presents the causality among interested $w=(2p+s+2)$ variables, where the node set $\bm{D}$ can be described using a linear SEM with an unknown $\bm{B}$ and some noise vector $\bm{\epsilon}$ as 
\begin{equation}
\resizebox{0.9\hsize}{!}{$
    \label{LSEM}
    \bm{D} = \bm{B}^\top \bm{D}+\bm{\epsilon}  :\quad  \begin{matrix}
p\\1\\p\\s\\1
\end{matrix}\begin{bmatrix}\bm{X}\\A\\\bm{X}A\\\bm{M}\\Y\end{bmatrix} = \bm{B}^\top\begin{bmatrix}\bm{X}\\A\\\bm{X}A\\\bm{M}\\Y\end{bmatrix} +\begin{bmatrix} \bm{\epsilon_X}\\\epsilon_A\\\bm{\epsilon}_{\bm{X}A}\\\bm{\epsilon_M}\\ \epsilon_Y\end{bmatrix},
$}
\end{equation}
where  $\bm{\epsilon_X},\epsilon_A,\bm{\epsilon_M}, \epsilon_Y$ are random error variables associated with $\bm{X}, A, \bm{M},$ and $Y$, respectively, and are jointly independent with mean zero. Note that {the noise is distribution-free}, so our method can handle any arbitrary distribution without breaking the linear SEM. In addition, since $\bm{X}A$ is determined given $\bm{X}$ and $A$, in order to accommodate the interaction setting for the linear SEM, we define $\bm{\epsilon}_{\bm{X}A} \equiv \bm{X}A$ with a constraint that $\bm{B}^\top \bm{D}$ produces a zero vector for elements corresponding to $\bm{X}A$. Alternatively, one can let $\bm{\epsilon}_{\bm{X}A} $ be a zero vector, with a  constraint that $\bm{B}^\top \bm{D}$ produces exact $\bm{X}A$. For simplicity, we use the first formulation. We next explicitly characterize the weighted adjacency matrix  $\bm{B}$ that satisfies Model (\ref{LSEM}) based on causal knowledge among $\bm{X}, A, \bm{M},$ and $Y$, as well as the interaction term $\bm{X}A$. Specifically, according to HCGs as illustrated in Figure \ref{fig1} and assumptions in Section \ref{sec:assump}, we have:
\begin{enumerate}
\vspace{-0.3cm}
    \item $\bm{X}$ has no parents, i.e.,\\ $ g_1(\bm{B})=\sum_{j=1}^p \sum_{i=1}^{2p+s+2} |b_{i,j}| = 0;$
    \vspace{-0.25cm}
    \item the only parents of $A$ are $\bm{X}$, i.e.,\\ $g_2(\bm{B})=\sum_{i=p+1}^{2p+s+2} |b_{i,{p+1}}|= 0; $
    \vspace{-0.25cm}
    \item $Y$ has no descendants, i.e.,\\ $g_3(\bm{B})=\sum_{i=1}^{2p+s+2} |b_{2p+s+2,i}| = 0;$ and
    \vspace{-0.25cm}
    \item the interaction $\bm{X}A$ also does not have parents, i.e.,\\ $g_4(\bm{B})=\sum_{j=p+2}^{2p+1}\sum_{i=1}^{2p+s+2} |b_{i,j}| = 0.$
    \vspace{-0.3cm}
\end{enumerate}
The conditions in $ g_1(\bm{B})$ to $ g_4(\bm{B})$ yield the following matrix  $\bm{B}^\top$ consisting of unknown parameters whose  sparsity is due to prior causal information:
\[\bm{B}^\top = \begin{bmatrix}
\bm{0}_{p\times p}&\bm{0}_{p\times 1}&\bm{0}_{p\times p}&\bm{0}_{p\times s}&\bm{0}_{p\times 1}\\
\bm{\delta_X}&0&\bm{0}_{1\times p}&\bm{0}_{1\times s}&0\\
\bm{0}_{p\times p}&\bm{0}_{p\times 1}&\bm{0}_{p\times p}&\bm{0}_{p\times s}&\bm{0}_{p\times 1}\\
\bm{B_X}^\top&\bm{\beta}_A&\bm{B}^\top_{\bm{X}A}&\bm{B_M}^\top&\bm{0}_{s\times 1}\\
\bm{\gamma_X}&\gamma_A&\bm{\gamma}_{\bm{X}A}&\bm{\gamma_M}&0
\end{bmatrix},\]
where $\bm{0}_{a\times b}$ is a $a\times b$ zero matrix/vector, and the parameters $\bm{\delta_X},\bm{B_X}^\top,$ and $\bm{\gamma_X}$ represent the influence of $\bm{X}$, on the treatment $A$, the mediators $\bm{M}$, and the outcome $Y$, respectively. Likewise, $\bm{\beta}_A$ and $\gamma_A$ represent the influence of $A$ on $\bm{M}$ and $Y$, respectively, and $\bm{\gamma_M}$ represent the influence of $\bm{M}$ on $Y$. $\bm{B_M}^\top$ represents the influence of the mediators on other mediators. If $\bm{B_M}^\top = \bm{0}_{s\times s}$ then we say that mediators are \textit{parallel}, otherwise they are \textit{sequentially ordered}. Finally, $\bm{B}^\top_{\bm{X}A}$ and $\bm{\gamma}_{\bm{X}A}$ represent the influences of the interaction between the possible moderators and the treatment, $\bm{X}A$, on $\bm{M}$ and $Y$. This gives the proposed model the capability to characterize moderation to allow heterogeneity and multiple sequentially ordered mediators to allow complexity. 
\begin{remark}\label{rem:whyinter}
    We extend the linear SEM by directly integrating moderation into the casual graph. This is done for computational purposes and easier interpretation. Moderation is essential in learning HCGs and quantifying heterogeneity in causal effects. Specifically, the interaction $\bm{X}A$ in  Model (\ref{LSEM}) allows causal effects to depend on the value of $\bm{X}$. 
    As will be seen in Theorem \ref{mainthm}, without interactions in the model, causal effects would be the same for all subjects within a population. Our model is thus flexible enough to handle heterogeneity and is general enough to cover other frameworks such as moderated mediation and mediated moderation as detailed in Appendix \ref{asec:med_model}. 
    We also allow a fully heterogeneous causal relationship among mediators, including changes in causal direction, by adding $\bm{XM}$ interactions to the model such that $\bm{X}$ would moderate the way mediators affect each other (see Section \ref{XMext}). 
    Also see further discussion on HCGs 
in Appendix \ref{asec:more_HCGs}. 
\end{remark}


\begin{remark}
The structural constraints in $g_1$ to $g_4$ can be modified or more can be added to account for prior knowledge. The more information is known about the variables, the more structural constraints can be added to improve the estimation of the unknown causal graph \citep{cai2020anoce}.
\end{remark}

\subsection{Generating Heterogeneous Causal Graphs
}\label{sec:subgraph}
We next show how to produce an HCG given matrix  $\bm{B}^\top$ and the value of moderators based on Model (\ref{LSEM}). To this end, we utilize the $do(\cdot)$ operator studied in \citet{pearl2000causality}, which simulates an intervention such that a variable can take a specific value irrespective of parent variable effects. Suppose we are interested in a group of subjects with $\bm{X} = \bm{x}$. By setting the value of $\bm{X}$ as $\bm{x}$ via $do(\cdot)$ operator,
all the edges that go to $\bm{X}$ will be eliminated, which yields a new linear SEM:  
\begin{align}\label{project_HCGs}
\bm{D}_{do(\bm{X=x})} &= \bm{B}_{do(\bm{X=x})}^\top \bm{D}_{do(\bm{X=x})}+\bm{\epsilon}' \to\\
    \begin{bmatrix}A\\\bm{M}\\Y\end{bmatrix} &=\begin{bmatrix} \bm{\delta_X x}  \\ \nonumber\bm{B_X}^\top\bm{x} \\ \bm{\gamma_Xx} \end{bmatrix} + \bm{B}_{do(\bm{X=x})}^\top\begin{bmatrix}A\\\bm{M}\\Y\end{bmatrix} +\begin{bmatrix}  \epsilon_A\\  \bm{\epsilon_M}\\  \epsilon_Y\end{bmatrix},
\end{align} 
 \[\text{where }~~~ \bm{B}_{do(\bm{X=x})}^\top = \begin{bmatrix}
0&\bm{0}_{1\times s}&0\\
\bm{\beta}_A + \bm{x}\bm{\beta}_{\bm{X}A}&\bm{B_M}^\top&\bm{0}_{s\times 1}\\
\gamma_A + \bm{x}\gamma_{\bm{X}A}&\bm{\gamma_M}&0
\end{bmatrix},~~~~~\]
characterizes the heterogeneous causal graph with respect to $\bm{X} = \bm{x}$, and the additional intercept is a constant effect due to the fixed baseline. Thus, any HCG given estimated parameters for $\bm{B}^\top$ and $\bm{X} = \bm{x}$ can be obtained without extra model training, as illustrated in Figure \ref{fig1}. 

Our proposed method first estimates a causal graph on the population level by including the treatment-and-moderator interaction in modeling to reflect heterogeneity, which allows us to project the estimated causal graph on the subgroup level to a certain value of moderator $\bm{X}$ through the do-operator as in \eqref{project_HCGs}. Under different values of $\bm{X} = \bm{x}$, the corresponding projected causal graphs are thus produced as the HCGs. Mathematically, it can be understood as $E(\mathcal{\widehat{G}})=E_{\bm{X}}\{E(\mathcal{\widehat{G}}|\bm{X} = \bm{x})\}=E_{\bm{X}}\{E(\mathcal{\widehat{G}}|do(\bm{X} = \bm{x})\}$, where $\mathcal{\widehat{G}}$ is the estimated causal graph on the population level, the first equation is due to the formula of iterative expectation, and the second equation is owing to the no unmeasured confounder assumption. Thus, 
the entire causal graph at the population level can be viewed as 
the weighted average of HCGs and the weights are the density of $\bm{X}$. 
In the case that the subgroup of interest is characterized only by a subset of variables, we can set the values of those variables as appropriate and average across the rest as done in the population-level case.
\section{Heterogeneous Causal Effects}\label{sec:HCTs}
In this section, we officially define heterogeneous causal effects (HCEs) 
and derive their theoretical properties.
\subsection{Heterogeneous Causal Effects of Treatment}
We focus on HCEs of treatment first. Following \citet{pearl2000causality}, we have the heterogeneous total effect (HTE) to be the total effect of the treatment on the outcome, the natural heterogeneous direct effect (HDE) to be the effect of the treatment on the outcome that is not mediated by mediators, and the natural heterogeneous indirect effect (HIE) to be the effect of the treatment on the outcome that is regulated by mediators, given the baseline information. 
\begin{definition}
\label{ind_caus_eff}
\textit{Heterogeneous Causal Effects of Treatment:}
\begin{align*}
    HTE(\bm{x}) = &\partial E\big\{Y|do(A=a),\bm{X}=\bm{x}\big\}/\partial a,\\
    HDE(\bm{x}) = &\partial E\big\{Y|do(A=a,\bm{M}=\bm{m}^{(a')}),\bm{X}=\bm{x}\big\}/\partial a,\\
    HIE(\bm{x}) = &\partial E\big\{Y|do(A=a',\bm{M}=\bm{m}^{(a)}),\bm{X}=\bm{x}\big\}/\partial a.
\end{align*}
\textit{Here, $\bm{m}^{(a)}$ is the value of $\bm{M}$ by setting $A = a$, and $a'$ is some fixed value.}
\end{definition}
In Definition \ref{ind_caus_eff}, the effects are functions of $\bm{x}$ as they depend on the particular individual's baseline information. Specifically, $HTE(\bm{x})$ is the change in the outcome of interest if we could have increased the treatment by one unit. Similarly, $HDE(\bm{x})$ can be interpreted as the change of the outcome due to one unit increase of the treatment when holding all mediators fixed, as the direct edge between $A$ and $Y$ in Figure \ref{fig1} given different $\bm{x}$. In contrast, $HIE(\bm{x})$ captures the indirect effect of the treatment on the outcome regulated by mediators as two indirect paths (i.e., $A\to M_1\to Y$ and $A\to M_1\to M_2 \to Y$) in Figure \ref{fig1}, under different $\bm{x}$.  Similar causal effects, dubbed interventional effects, were proposed by \citet{vansteelandt2017interventional}, yet, under a very specific causal graph with only two mediators, without considering $\bm{X}$ as possible moderators. In contrast, the above HCEs are defined for a much more general scenario with multiple mediators. We develop the explicit form of the proposed HCTs under Model (\ref{LSEM}) in the following theorem.
\begin{theorem}
\label{mainthm}
Under assumptions (A1-A3) and the model described in Equation (\ref{LSEM}), we have: \\
1). $HDE(\bm{x}) = \gamma_A + \bm{\gamma}_{\bm{X}A}\bm{x}$; \\
2). $HIE(\bm{x})=\bm{\gamma_M}(\bm{I}_s - \bm{B_M}^\top)^{-1}\big(\bm{\beta}_A + \bm{B}^\top_{\bm{X}A}\bm{x}\big);$  \\
 3). $ HTE(\bm{x})=HDE(\bm{x})+HIE(\bm{x}),$\
 \
where $\bm{I}_s$ is a $s\times s$ identity matrix and $\bm{x}$ is the value of $\bm{X}$.
\end{theorem} 
The proof of Theorem \ref{mainthm} is provided in Section \ref{asec:proof} in Appendix. We make a few remarks. First, results in 1) and 2) provide the exact form of HCEs as functions of a given $\bm{x}$, where the functions are determined by the unknown parameters in the weighted adjacency matrix  $\bm{B}$ in Model (\ref{LSEM}). Therefore, learning $\bm{B}$ is not only the key to recovering the underlying causal structure but also the essential middle step of estimating the HCEs of interest. Second, results in 3) produce an exact decomposition of the heterogeneous total causal effect as the summation of the natural heterogeneous direct and indirect effects. This result thus generalizes Section 5.1.3 in \citet{pearl2009causal} by considering extra baseline information. This can be seen more clearly by setting $\bm{X}=0$, in which case Pearl's equations will be generated. Third, it can be shown based on results in 1) to 3) that when $\bm{\gamma}_{\bm{X}A}$ and $\bm{B}^\top_{\bm{X}A}$ are all zeros, the HCTs do not depend on $\bm{x}$ anymore. In other words, without the interaction terms, the causal impacts for different sub-populations are homogeneous. This supports the statements in Remark \ref{rem:whyinter}.
\subsection{Heterogeneous Causal Effects of Mediators}
We next define the HCEs of a particular mediator  on the outcome. To this end, we define some intermediate quantities that deliver the partial causal impact of potential mediators on the outcome. Specifically, denote the effect of the treatment on mediator $M_i$ given $\bm{X}=\bm{x}$ as
\begin{align*}
    \Delta_i(\bm{x}) \equiv  \partial E\{M_i|do(A=a),\bm{X}=\bm{x}\}/ \partial a,
\end{align*}
which quantifies how the treatment causally influences a single mediator given $\bm{X}=\bm{x}$. Next, we conceptualize the conditional mean outcome when increasing mediator $M_i$ by 1 while fixing all other mediators as
\begin{align*}
\resizebox{1.03\hsize}{!}{$
    \Phi_i(a,\bm{x}) \equiv  E\{Y|do(A=a,M_i=m_i^{(a)}+1, \bm{\Omega}_i=\bm{o}_i^{(a)}),\bm{X}=\bm{x}\},$}
\end{align*}
where $\bm{\Omega}_i$ is a vector of mediators that do not include mediator $M_i$ and $\bm{o}_i^{(a)}$ is the value of $\bm{\Omega}_i$ by setting $A=a$. With these intermediate definitions, we define the natural heterogeneous direct mediation effect (HDM), the natural heterogeneous indirect mediation effect (HIM), and the natural heterogeneous total mediation effect (HTM) as follows for each mediator on the outcome of interest, given the baseline. 
\begin{definition}\label{HIM_HDM}
\textit{Heterogeneous Causal Effects of Mediator 
$M_i$:}
\begin{equation*}
\resizebox{1.03\hsize}{!}{$
\begin{split}
    HDM_i(\bm{x}) = \big[&\Phi_i(a,\bm{x})- E\{Y|do(A=a),\bm{X}=\bm{x}\}\big]  \times \Delta_i(\bm{x}),\\
    HIM_i(\bm{x}) = \big[&E\{Y|do(A=a,M_i=m_i^{(a)}+1),\bm{X}=\bm{x}\}\\&\   -\Phi_i(a,\bm{x}) \big] \times \Delta_i(\bm{x}),\\
    HTM_i(\bm{x}) = \big[&E\{Y|do(M_i=m_i+1),\bm{X}=\bm{x}\}\\&\  -E\{Y|do(M_i=m_i),\bm{X}=\bm{x}\}\big]\times \Delta_i(\bm{x}).
    \end{split}$}
\end{equation*}
\end{definition}
Similar to Definition \ref{ind_caus_eff}, the mediation effects are functions of $\bm{x}$. The product inside Definition \ref{HIM_HDM} is used to combine the causal effects of mediator $M_i$ on the outcome with the total effect of the treatment on this mediator $M_i$, i.e., $\Delta_i$. Here, $HDM_i(\bm{x})$ indicates the direct change of the outcome of interest if we could have increased a particular mediator by one unit while holding all its descendent/children mediators fixed. Likewise, $HIM_i(\bm{x})$ corresponds to the indirect effect of a particular mediator on the outcome regulated by its descendent/children mediators. Finally,  $HTM_i(\bm{x})$ characterizes the total influence of a particular mediator regardless of the format. Take mediator $M_1$ in Figure \ref{fig1} as an instance. The effect associated with the direct path $A\to M_1\to Y$ is $HDM_1(\bm{x})$ and  with the indirect path  $A\to M_1\to M_2 \to Y$ is $HIM_1(\bm{x})$,  where the summation of these two paths representing $HTM_1(\bm{x})$. Our definition also generalizes the natural causal effects for an individual mediator proposed by \citet{cai2020anoce} and \citet{chakrabortty2018inference} by incorporating  baseline information and handling heterogeneity. Based on these definitions, we establish the theoretical forms of the mediation effects below proved in Appendix \ref{asec:proof}.

\begin{algorithm}[!t]
\begin{algorithmic}
    \caption{Interactive  Structural Learning}\label{algo1}
   \STATE {\bfseries Input:} data $\bm{D}$, baseline dimension $p$, dimension of mediators $s$, tolerance for estimating causal skeleton $\delta_b$.
   \STATE {\bfseries Initialize:} $\widehat{\bm{B}}\empty^0\xleftarrow{}\bm{0}$; $\widehat{\bm{B}}\xleftarrow{}\bm{0}$
   \STATE \textbf{Estimate} $\widehat{\bm{B}}\empty^0$ using a causal discovery method of choice by incorporating $h_2$ into the loss function.
   \STATE \textbf{Compute} $\widehat{\bm{B}}\empty^b$, where $\widehat{\bm{B}}\empty^b_{i,j} = \widehat{\bm{B}}\empty^0_{i,j} > \delta_b$
   \FOR{$i=1$ {\bfseries to} $2p+s+2$}
   \STATE \textbf{Estimate} direct children of $i$th node, denoted $D_i$, by $\{D_j|\widehat{\bm{B}}\empty^b_{i,j} \neq 0\}$
   \STATE \textbf{Update} $i$th row of $\widehat{\bm{B}}$ using coefficients of fitted LASSO model using $\{D_j|\widehat{\bm{B}}\empty^b_{i,j} \neq 0\}$ as the predictors and $D_i$ as the response
   \ENDFOR
   \STATE \textbf{Compute} the HCEs and desired CIs using $\widehat{\bm{B}}$
   \STATE \textbf{Return} $\widehat{\bm{B}}$, HCEs, and CIs
\end{algorithmic}
\end{algorithm}

\begin{theorem}
\label{secthm}
Under assumptions (A1-A3) and  Model (\ref{LSEM}),\\ 
1a).$~HDM_i(\bm{x}) = \{\bm{\gamma_M}\}_i\{(\bm{I}_s - \bm{B_M}^\top)^{-1}\big(\bm{\beta}_A + \bm{B}^\top_{\bm{X}A}\bm{x}\big)\}_i;$\\
1b). $\sum_{i=1}^s HDM_i (\bm{x}) = HIE(\bm{x});$\\
2). $HIM_i(\bm{x})=HTM_i(\bm{x}) - HDM_i(\bm{x});$\\
3). $HTM_i(\bm{x})=HIE(\bm{x}) - HIE_{\mathbb{G}(-i)}(\bm{x}),$\\
where $\{\cdot\}_i$ is the $i$th element of a vector and $HIE_{\mathbb{G}(-i)}$ is the HIE under the causal graph $\mathbb{G}(-i)$ in which the $i$th mediator is removed from the original causal graph $\mathbb{G}$.
\end{theorem}
\vspace{-0.1cm}
We make a few remarks for Theorem \ref{secthm}. First, results in 1a) give the explicit form of the proposed HDM, where each multiplier corresponds to the $i$th element of the vector multiplier in results 2) of Theorem \ref{mainthm} for HIE. This immediately reveals the relationship between HIE and HDMs, i.e., the summation of all heterogeneous direct mediation effects equals to the heterogeneous indirect effect of the treatment, shown by results 1b) in Theorem \ref{secthm}. Second, the total mediation effect for a mediator $M_i$, $HTM_i$, can be interpreted as the effect of the treatment $A$ on the outcome $Y$ regulated by mediator $M_i$, or inversely as the change in total treatment effect that is due to $M_i$ being removed from the causal graph. This provides a feasible way to obtain $HTM$s as indicated in results 3). Third, according to Definition \ref{HIM_HDM}, $HTM_i$ can be further decomposed into direct effect that goes directly to the outcome through mediator $M_i$ from the treatment (i.e., $HDM_i$) and indirect effect that is regulated by other descendant mediators of $M_i$ as $HIM_i$. Such decomposition gives an alternative method to calculate $HIM_i$, as implied in results 2).
 \vspace{-0.2cm}
\section{Main Algorithm}
\label{sec:est_main}
 \vspace{-0.1cm}
In this section, we detail the procedure used to estimate HCGs and HCEs with three components. 

\textbf{1. Causal discovery with structural constraints.} In order to estimate the weighted adjacency matrix $\bm{B}$, we follow the approach taken by \citet{cai2020anoce} and combine the background structural knowledge with the score-based CSL \citep[e.g.,][]{ramsey2017million,zheng2018dags,yu2019dag,zhu2019causal}. This requires finding the best $\bm{B}$ given that it must be a DAG and satisfy the structural constraints posed in Section \ref{sec:lsem_intera}. In order for $\bm{B}$ to be a DAG, it can be shown that it must satisfy the following continuous acyclicity constraint typical of score-based CSL \citep[e.g.,][]{yu2019dag,zhu2019causal}: 
\[h_1(\bm{B}) = \text{tr}[(\bm{I}_{2p+s+2}+t\bm{B}\odot\bm{B})^{2p+s+2}] - (2p+s+2) = 0,\]
where $\text{tr} $ is the trace of a matrix, $t$ is a hyper-parameter that depends on an estimation of the largest eigenvalue of the matrix $\bm{B}$, and $\odot$ denotes the element-wise square. In order for $\bm{B}$ to satisfy  structural constraints, $g_1(\bm{B})$ to $ g_4(\bm{B})$ (see Section \ref{sec:lsem_intera}), it must satisfy: $
h_2(\bm{B}) =\sum_{i=1}^4g_i(\bm{B}) =0.$ \\
As remarked earlier, more structural constraints can be added and any added would be included in $h_2(\bm{B})$. Combining the two constraints above ($h_1$ and $h_2$) yields the following objective loss by an augmented Lagrangian 
\citep{cai2020anoce}, 
\begin{align*}
    L(\bm{B},\theta) &= f(\bm{B},\theta) + \lambda_1 h_1(\bm{B}) \\&\ +  \lambda_2 h_2(\bm{B}) + c|h_1(\bm{B})|^2 + d|h_2(\bm{B})|^2,
\end{align*}
where $f(\bm{B},\theta)$ is some loss function such as the Kullback-Leibler (KL) divergence in DAG-GNN  by  \citet{yu2019dag} (see details and choices of the base learner in Appendix \ref{asec:algo}), where model parameter $\theta$, $\lambda_1$ and $\lambda_2$ are Lagrange multipliers, and $c$ and $d$ are tuning parameters to ensure a hard constraint on $h_1$ and $h_2$. 
With these methods we can generate an estimate for $\bm{B}$, denoted $\widehat{\bm{B}}\empty^0$.

\textbf{2. Debiasing with regularized regression.} 
To reduce the bias in $\widehat{\bm{B}}\empty^0$ introduced during causal discovery, we apply a refitting procedure \citep{shi2021logan} in accommodating Model (\ref{LSEM}). First, we extract the estimated causal skeleton from $\widehat{\bm{B}}\empty^0$ by the binary matrix $\widehat{\bm{B}}\empty^b$ where $\widehat{\bm{B}}\empty^b_{i,j} =\mathbb{I} \{|\widehat{\bm{B}}\empty^0_{i,j}| > \delta\}$. Here, $\delta$ is a pre-determined threshold to prune edges in the estimated causal graph with weak signals. Applying a \textit{tolerance level} is a standard technique used in the causal discovery literature \citep[see more details in ][]{zheng2018dags,yu2019dag,zhu2019causal}. In practice, we choose the best $\delta$ to minimize the loss (such as the mean squared error) between the original data and the data generated by $\widehat{\bm{B}}\empty^0 \odot \mathbb{I} \{|\widehat{\bm{B}}\empty^0 | > \delta\}$. For each node $D_i$, we obtain its parent set as $PA_{\widehat{\mathbb{G}}}(D_i)$, where $\widehat{\mathbb{G}}$ is the estimated causal graph based on $\widehat{\bm{B}}\empty^b$, and regress, using LASSO, $D_i$ on $PA_{\widehat{\mathbb{G}}}(D_i)$ to obtain the refitted coefficients $\widehat{b}_{i,j}$ if $D_j\in PA_{\widehat{\mathbb{G}}}(D_i)$ and $\widehat{b}_{i,j} =0$ otherwise. Denote the refitted matrix as $\widehat{\bm{B}}$, which is used to estimate causal effects based on Theorems \ref{mainthm} and \ref{secthm}. While we agree that the regularized regression is biased, yet, with sufficient sample size, the bias term is negligible. Such a debiasing idea is commonly used for statistical inference in high-dimensional settings \citep{minnier2011per, zhang2014conf, ning2017gentheory}.

\textbf{3. Construction of confidence intervals (CIs).} We construct CIs of the estimated effects based on bootstrap \citep{efron1994introduction} and false discovery control \citep{benjamini2001control}. Specifically, we resample $n$ data points with replacements for independent $K$ times. In each bootstrap sample, we estimate the causal graph and compute the desired effect. For distribution-free noises, one can apply quantile-based CI by using $\alpha/2$-quantile of the bootstrap estimates as the lower bound and the $(1-\alpha/2)$-quantile as the upper bound, where $\alpha$ is the significance level. In the case that the data is Gaussian in nature, the Gaussian bootstrap CI can be applied. We name the combined procedure Interactive Structural Learning (ISL). The full algorithm is included in Appendix \ref{asec:algo}, however, an abbreviated version can be found in Algorithm \ref{algo1}, with the time complexity provided in Appendix \ref{asec:time}.


\vspace{-0.2cm}
\section{Extension to Functional SEM}
\label{ext}
We can extend our methodology to the \textit{nonlinear} setting by generalizing the nature of the relationship between the nodes and their parents. Specifically, for each node $Z_i\in\bm{Z}$, we model the relationship between this node and one or more of its parents as $\alpha f_{ij}\{PA^{j}_{{\mathbb{G}}}(Z_i)\}$ where $PA^{j}_{{\mathbb{G}}}(Z_i)$ is the $j$th parent of $Z_i$ and $\alpha$ is a scalar, vector, or matrix coefficient that is appropriately shaped serving the same purpose as the coefficients in $\bm{B}^\top$ from Model (\ref{LSEM}). Keep in mind, that the ``combined" variable of $\bm{X}$ and A is considered a single distinct from $\bm{X}$ and $A$ such that the heterogeneity of the data can be fully captured and interpreted. For example, the relationship between $Y$ and $A$ is modeled by $\gamma_A f_{Y,a}(A)$ and the modification of $A$'s relationship with $\bm{M}$ is modeled by $\bm{B^\top_{XA} f_{M,xa}(X},A)$. Putting these all together, we have the following functional SEM:
\begin{equation}\label{NonLin-SEM}
\resizebox{0.85\hsize}{!}{$
\left\{\begin{aligned}
\bm{X} &= \bm{\epsilon_X},\\
A &= \bm{\delta_X f_{A,x}(X)} + \epsilon_A,\\
\bm{M} &= \bm{B^\top_X f_{M,x}(X)} + \bm{\beta}^\top_A \bm{f_{M,a}}(A)\\&\ \ \ + \bm{B^\top_{XA} f_{M,xa}(X},A) + \bm{B^\top_M M} + \bm{\epsilon_M},\\
Y &= \bm{\gamma_X f_{Y,x}(X)} + \gamma_A f_{Y,a}(A)\\&\ \ \ + \bm{\gamma}_{\bm{X}A}f_{Y,xa}(\bm{X},A) + \bm{\gamma_M f_{Y,m}(M)} + \epsilon_Y.
\end{aligned}\right.$}
\end{equation}
We take a functional approach here, rather than a fully nonlinear approach, in order to have a closed form for the causal effects and allow for the model to be interpreted in a similar manner to the linear Model (\ref{LSEM}). This also allows us to directly extend the ISL algorithm in Section \ref{sec:est_main} 
to the functional case,  by replacing the base learner with a learner proven to work efficiently in the nonlinear setting and generalizing the regression step to be appropriate for any functional model. The base learner we have chosen to use is the one recently introduced by \citet{rolland2022score} which uses the data distribution's score function to prune the full DAG to estimate the causal graph. For more information see Appendix \ref{sec:append_algo2}.  
After the causal graph has been estimated, we can use it to estimate the coefficients in Model (\ref{NonLin-SEM}) by first positing a functional model and fitting it, given the estimated parents for each of the nodes. Here is where background information, like structural constraints, can be incorporated. While structural constraints can also be applied in the above ordering process, we have found that in practice it is more efficient to do so afterwards during the fitting process. Background information can also be incorporated in the form of the chosen functional models used for each relationship. Without background information, misspecification can occur, leading to inaccurate causal effect estimation.

\begin{algorithm}[!t]
\begin{algorithmic}
    \caption{Interactive Structural Learning Extended}\label{algo2}
   \STATE {\bfseries Input:} data $\bm{D}$, baseline dimension $p$, dimension of mediators $s$, tolerance for estimating causal skeleton $\delta_b$.
   \STATE {\bfseries Initialize:} $\widehat{\bm{B}}\empty^0\xleftarrow{}\bm{0}$; $\widehat{\bm{B}}\xleftarrow{}\bm{0}$
   \STATE \textbf{Estimate} the skeleton of $\widehat{\bm{B}}\empty^0$ using \citet{rolland2022score}'s SCORE method.
   \STATE \textbf{Fit} chosen functional model using estimated skeleton while incorporating structural constraints to estimate $\widehat{\bm{B}}\empty^0$
   \STATE \textbf{Compute} $\widehat{\bm{B}}\empty^b$, where $\widehat{\bm{B}}\empty^b_{i,j} = \widehat{\bm{B}}\empty^0_{i,j} > \delta_b$
   \FOR{$i=1$ {\bfseries to} $2p+s+2$}
   \STATE \textbf{Estimate} direct children of $i$th node, denoted $D_i$, by $\{D_j|\widehat{\bm{B}}\empty^b_{i,j} \neq 0\}$
   \STATE \textbf{Update} $i$th row of $\widehat{\bm{B}}$ using coefficients of fitted LASSO model using $\{D_j|\widehat{\bm{B}}\empty^b_{i,j} \neq 0\}$ as the predictors and $D_i$ as the response
   \ENDFOR
   \STATE \textbf{Compute} the HCEs and desired CIs using $\widehat{\bm{B}}$
   \STATE \textbf{Return} $\widehat{\bm{B}}$, HCEs, and CIs
\end{algorithmic}
\end{algorithm}

Once the functional model has been fitted, analogous theorems to those posed in Section \ref{sec:HCTs} can be used to calculate the causal effects. We emphasize here that the quantities defined in Definitions \ref{ind_caus_eff} and \ref{HIM_HDM} in Section \ref{sec:HCTs} are still applicable in the functional case, though the explicit form of heterogeneous causal effects derived in Theorems \ref{mainthm} and \ref{secthm} are somewhat different in the functional model setting. According to Definition \ref{ind_caus_eff}, we can develop the explicit form of the proposed HCEs for the treatment under Model (\ref{NonLin-SEM}) as follows. 
\begin{theorem}
\label{asec:mainthm}
Under assumptions (A1-A3) and Model (\ref{NonLin-SEM}),\\
1). $HDE(\bm{x}, a) = \gamma_A\frac{\delta f_{Ya}(a)}{\delta a} + \bm{\gamma}_{\bm{X}A}\frac{\delta f_{Yxa}(\bm{x},a)}{\delta a};$\\
2). $HIE(\bm{x}, a) = \bm{\gamma_M}\frac{\delta \bm{f_{Ym}}(m^{(a)})}{\delta \bm{m}^{(a)}}\frac{\delta \bm{m}^{(a)}}{\delta a};$\\
3). $HTE(\bm{x}, a)=HDE(\bm{x}, a)+HIE(\bm{x}, a);$\\
where $\bm{m}^{(a)} = E[\bm{M}|do(A=a),\bm{X}=\bm{x}].$
\end{theorem} 
\vspace{-0.1cm}
Similarly, the form of HCEs for mediators can be developed.
\begin{theorem}
\label{asec:secthm}
Under assumptions (A1-A3) and Model (\ref{NonLin-SEM}),\\
1a). $HDM_i(\bm{x}, a) = \{\bm{\gamma_M}\}_i \frac{\delta \bm{f_{Ym}(m^{(a)})}}{\delta m_i^{(a)}} \frac{\delta \bm{m_i}^{(a)}}{\delta a};$\\
1b). $\sum_{i=1}^s HDM_i (\bm{x},a) = HIE(\bm{x},a);$\\
2). $HIM_i(\bm{x},a)=HTM_i(\bm{x},a) - HDM_i(\bm{x},a);$\\
3). $HTM_i(\bm{x},a)=HIE(\bm{x},a) - HIE_{\mathbb{G}(-i)}(\bm{x},a).$
\end{theorem}
\vspace{-0.1cm}
Proofs are presented for both theorems in the appendix. The main results under Theorems \ref{asec:mainthm} and \ref{asec:secthm} can be interpreted analogously as in Theorems \ref{mainthm} and \ref{secthm}. The extended functional ISL is given in Algorithm \ref{algo2}. 

\section{Simulation Studies}\label{sec:simu}
We test the ability of our method on estimating HCGs and HCEs in a variety of situations. The computing infrastructure used is a virtual machine at Google Colab with the Pro tier for the majority of the computations, and a cluster server with 30 processor cores with 15GB of random access memory for computing bootstrap estimates in parallel. 

\textbf{Scenario generation.} We generate 8 scenarios. Scenario 1 (S1) is the simple case in which only $\bm{X}$, $A$, and $Y$ have non-zero weights; Scenario 2 (S2) has parallel mediators ($\bm{B_M}^\top = \bm{0}_{s\times s}$); and Scenario 3 (S3) has sequentially ordered mediators ($\bm{B_M}^\top \neq \bm{0}_{s\times s}$). The true causal graphs for S1 and S2 are randomly generated given structural constraints, while the causal graph for S3 is generated using the Erd\H os-Re\'nyi (ER) model with a degree of 2. The weights in graphs are randomly chosen from $\{-1,0,1\}$. Each of these scenarios sets $p=2$ and $s=6$ for a total of 12 nodes in the graph (including treatment, interaction, and outcome).
In addition, we require that there is at least one non-zero interaction term to ensure that causal graphs are heterogeneous as discussed in Remark \ref{rem:whyinter}. However, we also study two alternative cases: (1) a variant of S3 in which there is no interaction term (S3nx) and (2) a variant of S3 in which $\bm{X}$ is a moderator as defined by \citet{kraemer2002mediators} (S3mod), that is, $\bm{X}$ is independent of $A$ and there is an interaction term. Scenario 4 (S4), 5 (S5), and 6 (S6) are higher-dimensional versions of S3, with $s=38$ in S4, $p=18$ in S5, and $s=10$ plus $p=22$ in S6, where the underlying graphs are generated by the ER model with a degree of 4. The rest settings are the same as in S3. The sample size, $n$, is chosen from $\{500,1000\}$. For S4-S6, the sample size is fixed at $n=1000$. We also set the noise to be standard gaussian for $\bm{X}$, $A$, $\bm{M}$, and $Y$. Given the noise and the edge weights, Model \ref{LSEM} can be directly used to generate the data by setting $\bm{X}$ as its noise and using $\bm{X}$ to generate $A$ with the edge weights and associated noise. $\bm{X}$ and $A$ can then be used with the edge weights and associated noise to generate $\bm{M}$ which is in turn used to generate $Y$ in a similar manner. After this process is complete a baseline value of 1.0 is added to the response. Finally, the column-wise mean value for the dataset is subtracted from the dataset. The seeds used to generate the 100 datasets used for each simulation scenario are 1 through 100.

\begin{figure}[!t]
        \centering
        \begin{subfigure}{0.23\textwidth}
            \centering
            \includegraphics[width=1\textwidth]{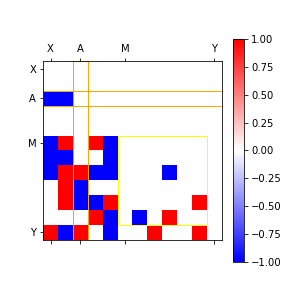}
            \caption{True}%
        \end{subfigure}
        \begin{subfigure}{0.23\textwidth}  
            \centering 
            \includegraphics[width=1\textwidth]{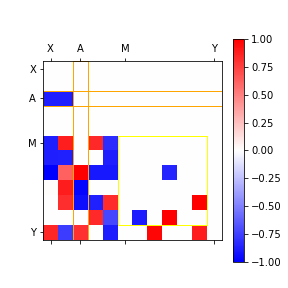}
            \caption{ISL($n=500$)}%
        \end{subfigure}
        \vskip\baselineskip
        \vspace{-0.4cm}
        \begin{subfigure}{0.23\textwidth}   
            \centering 
            \includegraphics[width=1\textwidth]{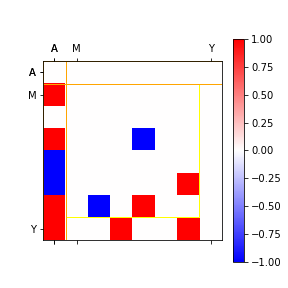}
            \caption{$\bm{X}=[1,0]$}%
        \end{subfigure}
        \begin{subfigure}{0.23\textwidth}   
            \centering 
            \includegraphics[width=1\textwidth]{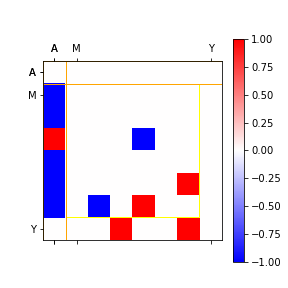}
            \caption{$\bm{X}=[0,1]$}%
        \end{subfigure}
         \vspace{-0.2cm}
        \caption{Results for S3 with subgraphs for different $\bm{X}$.}
        \label{S3sub}
\end{figure}
\vspace{-0.2cm}

\textbf{Simulation results.} We apply the proposed method to estimate the heterogeneous causal graphs and effects. The averaged estimated matrix $\widehat{\bm{B}}^\top$ over 100 replicates under different sample sizes with its true causal graph are summarized in Figure \ref{s1b} to \ref{s6b} in Appendix for S1 to S6 to evaluate the proposed method. Each plot has a yellow box that identifies the inner weighted adjacency matrix for mediators and intersecting orange lines to emphasize the direct effects on or from the treatment. The corresponding accuracy metrics, here we use the false discovery rate (FDR), the true positive rate (TPR), and the structural Hamming distance (SHD), with their standard errors are provided in Tables \ref{accs135h} and \ref{accs135h2}. The bias and standard error of the estimated heterogeneous causal effects for S1-S3mod are summarized in Tables \ref{bias135h}-\ref{biasnxmod} in the Appendix. The effects for S4-S6 are omitted for brevity. In addition, we compare our method with NOTEARS proposed by \citet{zheng2018dags} and DAG-GNN proposed by \citet{yu2019dag} for S1 to S3 with $n=500$, as summarized in Tables \ref{accs135n} and \ref{accs135d}. The penalty term in NOTEARS is set at $l_1=0.03$ the settings for DAG-GNN are left as default to achieve optimal performance. Here, for all estimated causal graphs, we set a threshold of 0.4 as the maximum absolute value that an edge weight needs to be to be considered nonzero. We also conduct a bootstrap simulation using 100 runs of 1000 bootstrap resamples to create 100 bootstrap CIs for each causal effect associated with S3. The average coverage probabilities are reported in Table \ref{boot} with $\alpha =0.05$. Here we note that in estimating graphs for use in constructing these confidence intervals, thresholding is used to accurately determine the causal skeleton. As will be seen in the real data section, thresholding can sometimes lead to cases in which all bootstrap samples return the same estimate yielding a confidence interval that only contains one value. While this is not ideal, it should be noted that this is not an error, but a consequence of using accuracy-improving techniques, such as thresholding.

As can be seen in Table \ref{accs135h} and \ref{accs135h2}, our method can estimate the weighted adjacency matrix of a causal graph with high accuracy which in turn leads to highly accurate estimates for the causal effects as seen in Tables \ref{bias135h}-\ref{biasnxmod}. The accuracy also increases marginally as the sample size increases. Our method can also be seen, in Figures \ref{s1b} to \ref{s3b} and Table \ref{accs135n}, to perform better than NOTEARS. Furthermore, our method can be seen to perform well in scenarios without interaction and in the case in which moderators that are independent of the treatment are used. In the higher-dimension case, it can be seen in Table \ref{accs135h2} and \ref{accs46} that the accuracy of the proposed method remains quite high. S5 had the best accuracy metrics of the three and S4 had the worst, with S6 in the middle. It can then be concluded that in general, it is easier to estimate a graph with more potential moderators than with potential mediators. Here, SHD should be interpreted given the total possible errors a model could make. Since S4 has 44 nodes, $44^2/2=968$ errors are possible. So an SHD of 14.67 for S4 would be equivalent to an accuracy of 98.5\%.

\section{Real Data Analysis}\label{sec:real}
 
To demonstrate the practical usefulness of our method, we apply it to investigate the causal relationship of psychiatric disorders for trauma survivors using data collected from the AURORA study. The response of interest is post-traumatic stress disorder (PTSD) measured three months after trauma. The event of interest is pre-trauma insomnia of trauma survivors. Peri-traumatic (PT) PTSD, stress, acute distress (ASD), and depression  measured 2 weeks after the trauma event are included as potential mediators. Moderators used for this study include age, gender, race, education, pre-trauma physical and mental health, perceived stress level, neuroticism, and childhood trauma. Before applying our method, categorical variables were one-hot encoded, numerical variables were centered, and missing data were removed. Results are summarized in causal graphs in Figures \ref{male_graph} - \ref{black_female_graph} with higher resolution for easier viewing. Motivations 
and overall effects are summarized in Appendix \ref{asec:real_des}.

\begin{table}[!t]
\centering
\small
    \captionof{table}{Accuracy metrics (standard errors) 
    for $n=500$.}
    \label{accs135h}
\begin{tabular}{clll}
  \hline
  & FDR & TPR & SHD\\  \hline
  S1 & 0.00(0.00) & 1.00(0.00) & 0.00(0.00)\\      \hline
  S2 & 0.01(0.02) &1.00(0.03) & 0.25(0.86)\\      \hline
  S3 & 0.00(0.01) & 1.00(0.01) & 0.04(0.24)\\      \hline
    S3nx & 0.00(0.01) & 1.00(0.01) & 0.03(0.22)\\      \hline
  S3mod & 0.00(0.00) & 1.00(0.01) & 0.02(0.14)\\      \hline
\end{tabular} 
\end{table}
\begin{table}[!t]
\centering
\small
    \captionof{table}{Accuracy metrics (standard errors) 
    for $n=1000$.}
    \label{accs135h2}
\begin{tabular}{clll}
  \hline
  & FDR & TPR & SHD\\  \hline
  S1 & 0.00(0.00) & 1.00(0.00) & 0.00(0.00)\\      \hline
  S2 & 0.00(0.00) & 1.00(0.01) & 0.01(0.10)\\      \hline
  S3 & 0.00(0.00) & 1.00(0.00) & 0.00(0.00)\\      \hline
  S3nx & 0.00(0.00) & 1.00(0.00) & 0.00(0.00)\\      \hline
  S3mod & 0.00(0.00) & 1.00(0.00) & 0.00(0.00)\\      \hline
  S4 & 0.03(0.01) & 0.90(0.02) & 14.67(3.45)\\      \hline
    S5 & 0.00(0.01) & 1.00(0.02) & 0.43(2.53)\\      \hline
  S6 & 0.00(0.01) & 0.99(0.02) & 2.39(6.26)\\      \hline
\end{tabular} 
\end{table}

\begin{figure}[!t]
   \centering
   \includegraphics[width=1\linewidth]{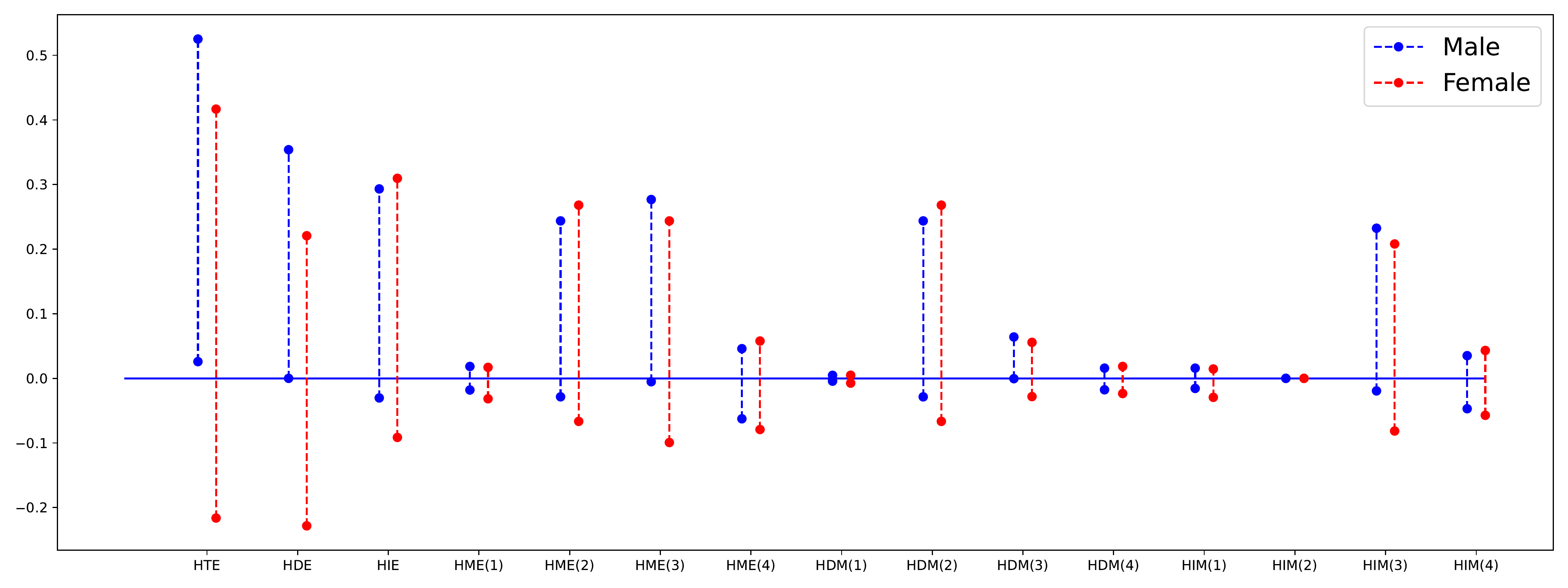}
   \caption{$95\%$ Confidence intervals for effects of male and female subgraphs, produced using 1000 bootstrap samples.} 
     \label{male_female_confint}
     \vspace{-0.2cm}
\end{figure}

Heterogeneous causal effects for this real data example are characterized by the interactions of pre-trauma insomnia with moderators, such as gender and race. These interaction terms can affect the effect of pre-trauma insomnia to response, or to potential mediators. Results are summarized in Figures \ref{resp_confint} - \ref{med4_confint}. Results in Figure \ref{med1_confint} suggest the interaction of pre-trauma insomnia with Gender, education, and lifetime chronic stress are marginally significant for 3-month PTSD which indicates that the direct effect of pretrauma insomnia to 3-month PTSD is moderated by gender, education level, and level of chronic stress. In Figure \ref{resp_confint} and \ref{med2_confint} to \ref{med4_confint}, we can also observe that the chronic stress level moderate (increase) the effect of pretrauma insomnia to peritraumatic distress, stress, PTSD, and depression.  On the other hand, education level decreases the effect of pretrauma insomnia on the four peritraumatic disorders. Finally, Figure \ref{med2_confint} suggest that gender moderate the effect of pretrauma insomnia to peritraumatic stress but not the other three peritraumatic disorders. Heterogeneous causal effects for male and female trauma survivors are shown in Figure \ref{male_female_confint}. Results suggest that the total, direct and indirect effects of pretrauma insomnia to 3-month PTSD are (marginally) significant for male trauma survivors but not for females, and the effect of pretrauma insomnia on 3-month PTSD are mainly mediated by peritrauma stress and PTSD. These results also indicate that preventive intervention of 3-month PTSD after trauma exposure that focuses on reducing pertrauma stress and PTSD is more likely to be effective for male than for female trauma survivors. The causal effects for the male and female subgraph are summarized in Tables \ref{male_tab} and \ref{female_tab}, respectively.

\subsubsection*{Acknowledgments}
Research reported in this publication was supported in part by the National Heart, Lung, And Blood Institute of the National Institutes of Health under Award Number T32HL079896. The content is solely the responsibility of the authors and does not necessarily represent the official views of the National Institutes of Health.

\bibliography{mycite}
\bibliographystyle{icml2023}

\newpage

\onecolumn
\appendix
\addcontentsline{toc}{section}{Appendix} 
\part{Appendix} 
\parttoc 
\counterwithin{figure}{section}
\counterwithin{table}{section}
\counterwithin{equation}{section}
\counterwithin{definition}{section}
\newpage
\section{More Details on Framework}

\subsection{Complete Definitions for Assumptions and further discussion}\label{asec:assump}
In order to understand the Causal Markov and Faithfulness assumptions, one must first understand the concept of D-separation.
\begin{definition}[D-separation]
Nodes, $X$ and $Y$, are d-separated by a set of nodes, $Z$, if and only if for every path, $p$, there exists a node, $m\in Z$, that extends $p$ ($i\to m\to j$) or forks $p$ ($i\xleftarrow{} m\to j$) and for any node, $c$, along $p$ that is a so-called collider ($i\to m\xleftarrow{} j$), $c$ and all descendents of $c$ are not in $Z$ \citep{pearl2009causal}.
\end{definition}
Given that $Z$ d-separates $X$ and $Y$ and $X$ preceeds $Y$ causally, the implication of d-seperation is that $X\independent Y|Z$.
\begin{definition}[Causal Markov assumption]
For a given causal graph, $G=\{Z,\bm{E}\}$, the set of independences among the nodes, $Z$, contains the set of independences implied by d-separation.
\end{definition}
\begin{definition}[Faithfulness assumption]
For a given causal graph, $G=\{Z,\bm{E}\}$, the set of independences among the nodes, $Z$, is \textbf{exactly} described by the set of independences implied by applying d-separation to $G$.
\end{definition}


It should be noted that the assumptions made in this paper are common in the literature of causal inference. Please refer to \cite{pearl2000causality, pearl2009causal, athey2015machine, nandy2017estimating, wager2018estimation, kunzel2019metalearners, nie2021quasi} for discussions of these assumptions and their impact. There are a few future extensions to relax or diagnose these assumptions. For instance, a full sensitivity analysis of the assumptions would be useful to the field when it is hard to include all variables causally related to any variable in the data in practice. In addition, utilizing the instrumental variables in the contexts of causal graphs with multiple mediators may be beneficial to address no unmeasured confounders. We leave these directions for future work.

\subsection{Further Discussion Regarding the Identifiability of the Causal Graph}\label{sec:identif}

Under the linear SEM, when the noises are Gaussian distributed, the model yields the class of standard linear-Gaussian model that has been studied in \citet{spirtes2000constructing, peters2017elements}. If the noises have equal variances, the DAG is uniquely identifiable from the observational data according to \citet{peters2014identifiability}. If the functions are linear but the noises are non-Gaussian, the true DAG can be uniquely identified under favorable conditions described in \citet{shimizu2006linear}. Also, the DAG can be naturally identified from the observational data if the corresponding MEC contains only one DAG, however, this cannot be known. The recent score-based causal discovery methods \citep{zheng2018dags, yu2019dag, zhu2019causal, cai2020anoce} usually consider synthetic datasets that are generated from fully identifiable models so that it is practically meaningful to evaluate the estimated graph with respect to the true DAG. The non-linear case shares similar identifiability statements, with more details provided in Sections 2.1 and 3.3 in \cite{zheng2020learning}.

\subsection{Comparison between Moderated Mediation and Mediated Moderation}\label{asec:med_model}
The interaction seen in Model (\ref{LSEM}) is an example of moderation in which the relationship between two variables depends on another. In this case, the moderator would be $\bm{X}$, and the relationships being moderated would be $A$ and $\bm{M}$ and $A$ and $Y$. For those familiar with moderation, our setting may seem similar to the moderated mediation and/or mediated moderation settings. In Figure \ref{modmed}, we have an example of moderation, moderated mediation, and mediated moderation taken from \citet{muller2005modmed}'s paper on the topic. 
Here $X$ is the treatment, $Me$ is the mediator, $Y$ is the response and $Mo$ is the moderator, not shown in the figure. Here it can be seen that $Mo$ moderates the relationship between $X$ and $Y$, which is the example of moderation described earlier. $Mo$ also moderates the relationship between $X$ and $Me$ and $Me$ and $Y$ which is an example of mediated moderation and moderated mediation respectively.

\begin{figure}[!htb]
  \centering
  \includegraphics[width=0.3\linewidth]{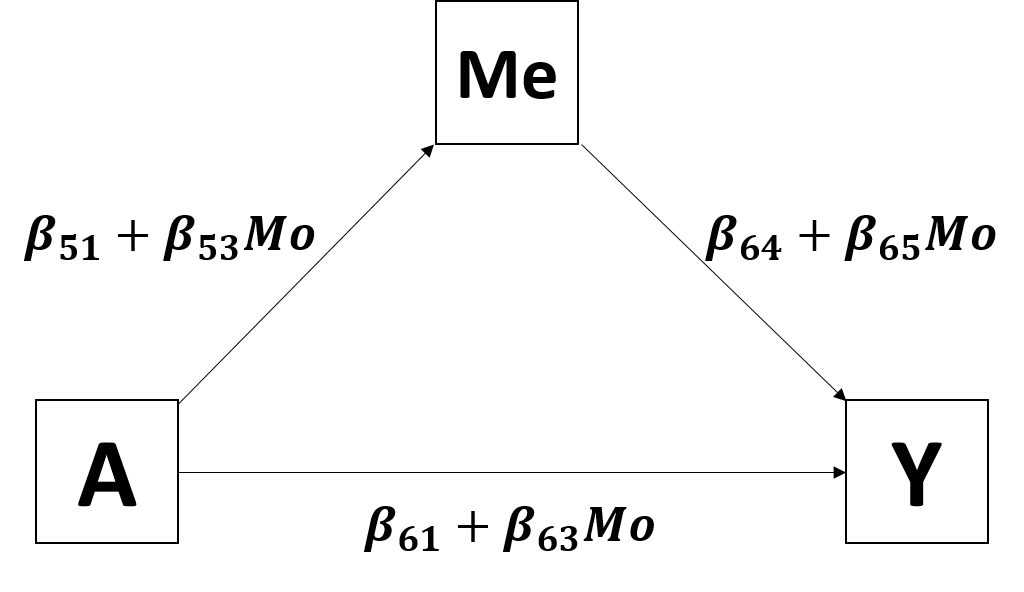} 
  \caption{Illustration of moderation, moderated mediation, and mediated moderation.}
\label{modmed}  
 \end{figure}
In their paper, \citet{muller2005modmed} give equations similar to the linear SEM equations above to model both of these settings:
\begin{align*}
    Me &= \beta_{50} + \beta_{51}X + \beta_{52}Mo + \beta_{53}XMo + \epsilon_5,\\
    Y &= \beta_{60} + \beta_{61}X + \beta_{62}Mo + \beta_{63}XMo + \beta_{64}Me + \beta_{65}MeMo + \epsilon_6.
\end{align*}
If we restate this as a linear SEM, the difference between our setting and the moderated settings can be seen:
\begin{equation}
    \begin{bmatrix}Mo\\X\\XMo\\Me\\MeMo\\Y\end{bmatrix} = \begin{bmatrix}0\\0\\0\\\beta_{50}\\0\\\beta_{60}\end{bmatrix} + \begin{bmatrix}
0&0&0&0&0&0\\
0&0&0&0&0&0\\
0&0&0&0&0&0\\
\beta_{52}&\beta_{51}&\beta_{53}&0&0&0\\
0&0&0&0&0&0\\
\beta_{62}&\beta_{61}&\beta_{63}&\beta_{64}&\beta_{65}&0
\end{bmatrix}\begin{bmatrix}Mo\\X\\XMo\\Me\\MeMo\\Y\end{bmatrix} + \bm{\epsilon}.
\end{equation}

Comparing this model to our model, we can see that our model is more general. The moderator in this model is allowed to directly affect the treatment in our model, that is to say, we allow the confounder to moderate relations between other variables.

\subsection{Further Discussion on Heterogeneous Causal Graphs}\label{asec:more_HCGs}

By heterogeneity, we are referring to the differences that may exist between subgroups within a population. We pose these differences via the way in which moderators affect and interact with the treatment variable in the causal graph. This leads to the possibility that different subgroups may have different causal graphs. This set of causal graphs is what we refer to as heterogeneous causal graphs (HCGs). Note that we assume here for practical reasons that $\bm{X}$ is a set of pre-treatment confounders and the topological order of $A$ is lower than that of $Y$, i.e. the causal direction of the edge between $A$ and $Y$ is assumed. Furthermore, the existence of any other edge is not assumed and may be different among the HCGs depending on how $\bm{X}$ moderates $A$. The model presented in this work, given the topological assumption from earlier, will not lead to HCGs with the same edge that differ in causal direction, except that under a certain value of $\bm{x}$ the edge from $A$ to $\bm{M}$ and $A$ to $Y$ may disappear based on the results in section \ref{sec:subgraph}. Such changes in causal edges from the treatment are owing to the $\bm{X}A$ interactions we considered. 

Following a similar logic, theoretically, we could allow a fully heterogeneous causal relationship among mediators, including changes in causal direction, by adding $\bm{XM}$ interactions to the model such that $\bm{X}$ would moderate the way mediators affect each other (see Section \ref{XMext}). A “fully heterogeneous” model may not be practically necessary since it is very uncommon to expect the causal direction to be heterogeneous. The slightly restricted version we proposed in this study should be able to address the majority of applied research questions with respect to heterogeneous causal relationships.


\subsection{Adding XM Interaction to Model \ref{LSEM}}\label{XMext}
In this section, we introduce an extension of Model \ref{LSEM} that includes $\bm{XM}$ interaction. Note that this interaction only includes a term for the response and does not include a term for mediators, i.e. we assume that $\bm{X}$ does not moderate how mediators affect other mediators. This is done mainly for three practical reasons: Firstly, we believe practitioners may not find value in modeling this relationship. Secondly, modeling this relationship dramatically increases the number of quantities that need estimating and thus increases the estimation error. Thirdly, including this term would not allow for a closed-form formula for the effects as we have done previously for Model \ref{LSEM} which reduces its interpretability.

\[\begin{matrix}
p\\1\\p\\s\\s*p\\1
\end{matrix}\begin{bmatrix}\bm{X}\\A\\\bm{X}A\\\bm{M}\\\bm{X\otimes M}\\Y\end{bmatrix} = \begin{bmatrix}
\bm{0}&\bm{0}&\bm{0}&\bm{0}&\bm{0}&\bm{0}\\
\bm{\delta_X}&0&\bm{0}&\bm{0}&\bm{0}&0\\
\bm{0}&\bm{0}&\bm{0}&\bm{0}&\bm{0}&\bm{0}\\
\bm{B^\top_X}&\bm{\beta_A}&\bm{B^\top_{XA}}&\bm{B^\top_M}&\bm{0}&\bm{0}\\
\bm{0}&\bm{0}&\bm{0}&\bm{0}&\bm{0}&\bm{0}\\
\bm{\gamma_X}&\gamma_A&\bm{\gamma_{XA}}&\bm{\gamma_M}&vec(\bm{\Gamma_{XM}})^\top&0
\end{bmatrix}\begin{bmatrix}\bm{X}\\A\\\bm{X}A\\\bm{M}\\\bm{X\otimes M}\\Y\end{bmatrix} + \begin{bmatrix}\bm{\epsilon_X}\\\epsilon_A\\\bm{X}A\\\bm{\epsilon_M}\\\bm{X\otimes M}\\\epsilon_Y\end{bmatrix}.
\]

\newpage
\begin{theorem}\label{app_thm_XM}
Under the modified assumptions and model above:
\begin{itemize}
    \item The natural direct effect \[DE = \gamma_A + \bm{\gamma_{XA} x};\]
    \item The natural indirect effect is \[IE=(\bm{\gamma_M} + \bm{x}^\top\bm{\Gamma_{XM}})\bm{(I_s - B^\top_M)^{-1}}\big(\bm{\beta_A} + \bm{B^\top_{XA} x}\big);\]
    \item The total effect of $A$ on $Y$ given $X$ is \[TE=\gamma_A + \bm{\gamma_{XA} x} + (\bm{\gamma_M} + \bm{x}^\top\bm{\Gamma_{XM}})\bm{(I_s - B^\top_M)^{-1}}\big(\bm{\beta_A} + \bm{B^\top_{XA} x}\big),\]
\end{itemize}
where $\bm{\Gamma_{XM}}$ is a $p$ by $s$ coefficient matrix and $vec(\bm{\Gamma_{XM}})^\top$ is the vectorization of this matrix.
\end{theorem}

As one might expect, when taking $\bm{XM}$ interaction into account, the indirect effect of $A$ changes accordingly.

\section{More Details of Algorithms}
\subsection{ISL Algorithm Specifics}\label{asec:algo}
In this section, we provide more details for the proposed method we call Interactive Structural Learning (ISL). We start with the extension of the base learner DAG-GNN \citet{yu2019dag} in our proposed framework with confounder-based interaction, with the corresponding implementation details below. The complete pseudocode is also provided below. Other score-based algorithms are also applicable to our framework such as causal discovery with RL \citep{zhu2019causal} expect the NOTEARS method \citep{zheng2018dags} which requires differentiable constraints to support their optimization algorithm.

We use \citet{yu2019dag}'s deep generative approach for causal structural learning as an example to detail the estimation of $\bm{B}$. First, we convert Model (\ref{LSEM}) and treat the random error as independent latent variables such that 
\begin{equation*}
\bm{D} = (\bm{I}_{2p+s+2}-\bm{B}^\top)^{-1}\bm{\epsilon}.
\end{equation*}
We adopt the variational autoencoder model in \citet{yu2019dag} to generate $\bm{D}$ using two multilayer perceptrons (MLPs) as the encoder and the decoder with mean zero noise $\bm{\epsilon}$. The goal is to minimize the evidence lower bound $f$ such that the generated data is close to the observed data as follows,
\begin{align*}
  f(\bm{B},\theta|\bm{D})  = {\sum_{i=1}^{2p+s+2}\Delta_{KL}\{q(\bm{\epsilon}|\bm{D}_i)\|p(\bm{\epsilon})\}}/{(2p+s+2)}  - E_{q(\bm{\epsilon}|\bm{D}_i)}\{\log p(\bm{D}_i|\bm{\epsilon})\}, 
\end{align*}
where $\bm{D}_i$ is the $i$-th variable in $\bm{D}$, $\theta$ is a vector of weights for MLPs, $\Delta_{KL}$ refers to the KL divergence, $q(\bm{\epsilon}|\bm{D}_i)$ is the reconstructed posterior distribution, $p(\bm{\epsilon})$ is the prior distribution, and $p(\bm{D}_i|\bm{\epsilon})$ is the likelihood function. This objective can be minimized using any preferred black-box stochastic optimization.

The pseudocode of Interactive Structural Learning (ISL) is provided in Algorithm \ref{algo1_full}.

\begin{remark}
Although ISL is able to identify causal heterogeneity, to have a complete examination of such causal heterogeneity, one has to apply the method to all possible configurations of the moderators. This may be a limitation in an extremely high dimensional setting, though we demonstrate a reasonably good performance of our method in Section \ref{sec:simu} with $p=22$ and more than $50$ nodes in total. We leave the variable selection in causal graphs as a future investigation.
\end{remark}

\subsection{Time Complexity of ISL Algorithm}\label{asec:time}

Our method uses another causal discovery method as a base in order to find the estimated topological ordering of the nodes. Due to this, the time complexity of our method is based on the base method used. Let the time complexity of the base method be $g(N,n)$ where $N$ is the data sample size, and the number of nodes, $n = 2p + s + 2$ where $p$ is the number of confounders/moderators and $s$ is the number of mediators. Let the time complexity of LASSO be $f(N,n)$. Computing the HCEs requires inverting a size $n$ matrix and a size $n-1$ matrix $s$ times. Putting this together, the time complexity of our method is $O(g(N,n) + f(N,n) + n^2)$ for $s<<n$ and $O(g(N,n) + f(N,n) + sn^2)$ otherwise.

 \newpage
\subsection{Full Algorithms}

\begin{algorithm}[!htb]
\begin{algorithmic}
    \caption{Interactive  Structural Learning}\label{algo1_full}
   \STATE {\bfseries Input:} data $\bm{D}$, dimension of baseline information $p$, dimension of mediators $s$, number of epochs $H$, max number of iterations $K$, parameter upper bound $U$, original learning rate $r_0$, tolerance for constraint $\delta$, tuning parameters $\rho$ and $\tau$, tolerance for estimating causal skeleton $\delta_b$.
   \STATE {\bfseries Initialize:} $\lambda_1\xleftarrow{}0$; $\lambda_2\xleftarrow{}0$; $c\xleftarrow{}1$; $d\xleftarrow{}1$; $r\xleftarrow{}r_0$; $\widehat{\bm{B}}\empty^0\xleftarrow{}\bm{0}$; $\widehat{\bm{B}}\xleftarrow{}\bm{0}$; $h_1^{old}\xleftarrow{}\infty$; $h_2^{old}\xleftarrow{}\infty$; Encoder weights $W^1$, $W^2$ and Decoder weights $W^3$, $W^4$ with Gaussian initialization.
   \FOR{$k=0$ {\bfseries to} $K-1$}
   \REPEAT
   \FOR{$i=0$ {\bfseries to} $H-1$}
   \STATE {\bfseries 1.} Compute the distribution parameters for the noise using the encoder: \[(\bm{\mu}_\epsilon,\bm{\sigma}_\epsilon)\xleftarrow{}(\bm{I}_{2p+s+2}-(\widehat{\bm{B}}\empty^0)^\top)MLP\{\bm{D},W^1,W^2\}\]
   \STATE {\bfseries 2.} Compute the distribution parameters for the data using the decoder:
   \[(\bm{\mu}_{\bm{D}},\bm{\sigma}_{\bm{D}})\xleftarrow{}MLP\{(\bm{I}_{2p+s+2}-(\widehat{\bm{B}}\empty^0)^\top)^{-1}\bm{\epsilon},W^3,W^4\}\]
   \STATE {\bfseries 3.} Calculate structural constraints, $h_1^{new}\xleftarrow{}h_1(\widehat{\bm{B}}\empty^0)$, $h_2^{new}\xleftarrow{}h_2(\widehat{\bm{B}}\empty^0)$, and loss $\ell\xleftarrow{}L(\widehat{\bm{B}}\empty^0,W^1,W^2,W^3,W^4)$
   \STATE {\bfseries 3.} Use backward propagation to update $\widehat{\bm{B}}\empty^0$ and network weights $\{W^1,W^2,W^3,W^4\}$
   \STATE {\bfseries 4.} Update learning rate $r$
   \ENDFOR
   \IF{$h_1^{new} > \rho h_1^{old}$ \& $h_2^{new} > \rho h_2^{old}$}
   \STATE $c\xleftarrow{}c\times\tau$; $d\xleftarrow{}d\times\tau$
   \ELSIF{$h_1^{new} > \rho h_1^{old}$ \& $h_2^{new} < \rho h_2^{old}$}
   \STATE $c\xleftarrow{}c\times\tau$
   \ELSIF{$h_1^{new} < \rho h_1^{old}$ \& $h_2^{new} > \rho h_2^{old}$}
   \STATE $d\xleftarrow{}d\times\tau$
   \ELSE
   \STATE Break
   \ENDIF
   \STATE {\bfseries Update:}  $h_1^{old}\xleftarrow{}h_1^{new}$; $h_2^{old}\xleftarrow{}h_2^{new}$; $\lambda_1\xleftarrow{}\lambda_1\times h_1^{new}$; $\lambda_2\xleftarrow{}\lambda_2\times h_2^{new}$;
   \IF{$h_1^{new} < \delta$ \& $h_2^{new} > \delta$}
   \STATE Break
   \ENDIF
   \UNTIL{$c\times d < U$}
   \ENDFOR
   \STATE \textbf{Compute} $\widehat{\bm{B}}\empty^b$, where $\widehat{\bm{B}}\empty^b_{i,j} = \widehat{\bm{B}}\empty^0_{i,j} > \delta_b$
   \FOR{$i=1$ {\bfseries to} $2p+s+2$}
   \STATE \textbf{Estimate} direct children of $i$th node, denoted $D_i$, by $\{D_j|\widehat{\bm{B}}\empty^b_{i,j} \neq 0\}$
   \STATE \textbf{Update} $i$th row of $\widehat{\bm{B}}$ using coefficients of fitted LASSO model using $\{D_j|\widehat{\bm{B}}\empty^b_{i,j} \neq 0\}$ as the predictors and $D_i$ as the response
   \ENDFOR
   \STATE \textbf{Return} $\widehat{\bm{B}}$
\end{algorithmic}
\end{algorithm}


\subsection{More Information on \citet{rolland2022score}'s SCORE Method and Extended ISL}\label{sec:append_algo2}

\citet{rolland2022score}'s SCORE Method is based on their finding that leaf nodes can be found by using the Jacobian of the data distribution's score function or the variance of the Jacobian. This fact is then used to estimate a topological ordering for the variables by iteratively finding a leaf node, adding it to the top of the current topological order, removing it from the graph, and repeating until all nodes have been ordered \cite{rolland2022score}. The causal graph is then estimated by using any pruning method to prune the full graph associated with the estimated topological order \cite{rolland2022score}. The SCORE algorithm is summarized below in Algorithm \ref{algo_score}. Note that for practical reasons, the authors identify the leaf node as the node corresponding as the smallest element of the variance of the Jacobian \cite{rolland2022score}. For further information on the computation of the score function, please see \citet{rolland2022score}'s paper.

\begin{algorithm}[!htb]
\begin{algorithmic}
    \caption{SCORE method}\label{algo_score}
   \STATE {\bfseries Input:} data $\bm{X}$ with $d$ nodes and sample size $N$.
   \STATE {\bfseries Initialize:} topological order $\pi=[]$ and nodes$=\{1,...,d\}$
   \FOR{$i=1$ {\bfseries to} $d$}
   \STATE \textbf{Estimate} score function $s_nodes$
   \STATE \textbf{Estimate} variance $V_j=\text{Var}_{X_{\text{nodes}}}\big[\frac{\delta s_j(X)}{\delta x_j}\big]$
   \STATE \textbf{Find} the leaf node: $l \xleftarrow{} \text{nodes}[\text{arg min}_j V_j]$
   \STATE \textbf{Add} leaf node to topological order: $\pi\xleftarrow{}[l,\pi]$
   \STATE \textbf{Remove} $l$ from the node set and data matrix
   \ENDFOR
   \STATE \textbf{Prune} the full DAG associated with $\pi$ using desired method
   \STATE \textbf{Return} the pruned DAG
\end{algorithmic}
\end{algorithm}

 \section{More Information about the Real Data Analysis}\label{asec:real_des}

 Heterogeneity is common for many types of disorders in terms of phenotype, risk factors, treatment effect, and underlying neuro-biological mechanism. This is especially true for psychiatric disorders \citep{marquand2016beyond}, such as depression and post-traumatic stress disorder (PTSD).
Heterogeneity is one of the important obstacles that continue to hamper research for psychiatric disorders \citep{marquand2016beyond}. 
To help overcome these challenges and advance research for psychiatric disorders, the National Institutes of Mental Health, joined by several institutions and foundations, launched 
the Advancing Understanding of RecOvery afteR traumA (AURORA) study \citep{mclean2020aurora}, where thousands of trauma survivors were recruited at the emergency departments after trauma exposures, and followed for one year to collect a broad range of bio-behavioral data. 

One of the aims of the AURORA study is to address the heterogeneity issue by identifying homogeneous sub-types or subgroups of patients after trauma exposure (i.e., homogeneous prognosis patterns among patients with PTSD). Currently, most of the researchers for psychiatric disorders focus on stratifying patients into subgroups using different clustering algorithms \citep{marquand2016beyond} based on cross-sectional or longitudinal patterns using self-reported data and neuroimaging data. On the other hand, a large number of studies have also been conducted to investigate the associations among pre-trauma characteristics, peritraumatic symptoms, and post-traumatic disorders. However, due to the lack of appropriate statistical tools, very few studies have been conducted to explore heterogeneous causal relationships for psychiatric disorders. 
 
 To develop preventive interventions for psychiatric disorders, such as post-traumatic stress disorder (PTSD), in the immediate aftermath of trauma exposure, it is essential to identify risk factors that can be targeted to prevent or alleviate the corresponding disorder. For instance, a recent study by \citet{neylan2021prior} found that pre-trauma sleep disorders (e.g., nightmare and insomnia) are associated with 8-week PTSD, and these associations are mediated by peritraumatic outcomes (e.g., acute stress disorder) measured at 2 weeks. However, as the authors have discussed, the associations identified in this study do not support causal conclusions. Furthermore, the possible causal effect of pre-trauma sleep disorder on PTSD is very likely moderated by characteristics such as gender, race, pre-trauma mental health, etc. As a result, it is valuable to understand the heterogeneous causal pathways of pre-trauma characteristics with PTSD and other psychiatric disorders. 

Under our linear model setting, the marginal effect of each covariate is directly characterized by the coefficient corresponding to the interaction term of each covariate with treatment $A$. 
Our model can help identify real moderators from a set of potential covariates which can be potentially very useful for applied studies since prior knowledge about real moderators is often not available.

 We worked with domain experts to try to take into account as many possible confounders as possible based on previous research, but it is impossible to claim that all potential confounders have been identified and included in the model. Results from this real data analysis indicate that overall the total, direct and indirect effects of pre-trauma insomnia are all marginally significant (HTE, HDE, and HIE in Figure \ref{whole_confint}), and total mediation effects from peritraumatic acute stress and peritraumatic PTSD are also marginally significant (HME(2) and HME(3) in Figure \ref{whole_confint}). The mediation effect from peritraumatic acute stress is mainly from the direct mediation effect (HDM(2) in Figure \ref{whole_confint}). On the other hand, the mediation effect from peritraumatic PTSD is from indirect mediation effects (HIM(3) in Figure \ref{whole_confint}). The causal effects for the whole population graph are summarized in Table \ref{whole_tab}.

\section{Additional Simulation Figures and Tables}\label{asec:simu}

\begin{table}[!htb]
    \centering
    \captionof{table}{Bias (standard error) of causal effects for S1-S3.}
    \label{bias135h}
\begin{tabular}{cllllll}
  \hline
  \multirow{2}{*}{ } 
      & \multicolumn{3}{c}{$n=500$} & \multicolumn{3}{c}{$n=1000$} \\             \cline{2-7}
  & S1 & S2 & S3 & S1 & S2 & S3 \\  \hline
  $HDE$ & -0.01(0.05) & -0.25(0.23) & -0.13(0.07) & -0.02(0.04) & -0.29(0.09) & -0.13(0.04)\\      \hline
  $HIE$ & 0.00(0.00) & 0.40(0.17) & 0.29(0.23) & 0.00(0.00) & 0.42(0.11) & 0.25(0.06)\\      \hline
  $HDM_1$ & 0.00(0.00) & 0.53(0.14) & 0.00(0.00) & 0.00(0.00) & 0.52(0.06) & 0.00(0.00)\\      \hline
  $HDM_2$ & 0.00(0.00) & 0.00(0.03) & 0.00(0.00) & 0.00(0.00) & 0.00(0.03) & 0.00(0.00)\\      \hline
  $HDM_3$ & 0.00(0.00) & 0.55(0.08) & 0.13(0.20) & 0.00(0.00) & 0.56(0.05) & 0.10(0.05)\\      \hline
  $HDM_4$ & 0.00(0.00) & -0.55(0.09) & 0.00(0.02) & 0.00(0.00) & -0.54(0.05) & 0.00(0.00)\\      \hline
  $HDM_5$ & 0.00(0.00) & -0.12(0.07) & 0.00(0.00) & 0.00(0.00) & -0.11(0.05) & 0.00(0.00)\\      \hline
  $HDM_6$ & 0.00(0.00) & 0.00(0.04) & 0.16(0.07) & 0.00(0.00) & 0.00(0.02) & 0.15(0.04)\\      \hline
  $HIM_1$ & 0.00(0.00) & 0.00(0.00) & 0.00(0.00) & 0.00(0.00) & 0.00(0.00) & 0.00(0.00)\\      \hline
  $HIM_2$ & 0.00(0.00) & 0.00(0.00) & 0.00(0.00) & 0.00(0.00) & 0.00(0.00) & 0.00(0.00)\\      \hline
  $HIM_3$ & 0.00(0.00) & 0.00(0.00) & 0.00(0.00) & 0.00(0.00) & 0.00(0.00) & 0.00(0.00)\\      \hline
  $HIM_4$ & 0.00(0.00) & 0.00(0.00) & 0.01(0.06) & 0.00(0.00) & 0.00(0.00) & 0.00(0.00)\\      \hline
  $HIM_5$ & 0.00(0.00) & 0.00(0.00) & 0.12(0.09) & 0.00(0.00) & 0.00(0.00) & 0.10(0.03)\\      \hline
  $HIM_6$ & 0.00(0.00) & 0.00(0.00) & 0.01(0.05) & 0.00(0.00) & 0.00(0.00) & 0.00(0.00)\\      \hline
\end{tabular}
\end{table}

\begin{table}[!htb]
    \centering
    \captionof{table}{Bias (standard error) of causal effects for S3nx and S3mod.}
    \label{biasnxmod}
\begin{tabular}{cllll}
  \hline
   & \multicolumn{2}{c}{$n=500$} & \multicolumn{2}{c}{$n=1000$}\\             \cline{2-5}
  & S3nx & S3mod & S3nx & S3mod\\  \hline
  $HDE$ & -0.11(0.07) & 0.02(0.08) & -0.12(0.05) & 0.03(0.06)\\      \hline
  $HIE$ & 0.13(0.18) & 0.73(0.15) & 0.11(0.05) & 0.72(0.12)\\      \hline
  $HDM_1$ & 0.00(0.00) & 0.00(0.00) & 0.00(0.00) & 0.00(0.00)\\      \hline
  $HDM_2$ & 0.00(0.00) & 0.00(0.00) & 0.00(0.00) & 0.00(0.00)\\      \hline
  $HDM_3$ & -0.05(0.18) & 0.49(0.12) & -0.07(0.05) & 0.48(0.09)\\      \hline
  $HDM_4$ & 0.00(0.00) & 0.00(0.00) & 0.00(0.00) & 0.00(0.00)\\      \hline
  $HDM_5$ & 0.00(0.00) & 0.00(0.00) & 0.00(0.00) & 0.00(0.00)\\      \hline
  $HDM_6$ & 0.19(0.05) & 0.25(0.09) & 0.18(0.02) & 0.24(0.07)\\      \hline
  $HIM_1$ & 0.00(0.00) & 0.00(0.00) & 0.00(0.00) & 0.00(0.00)\\      \hline
  $HIM_2$ & 0.00(0.00) & 0.00(0.00) & 0.00(0.00) & 0.00(0.00)\\      \hline
  $HIM_3$ & 0.00(0.00) & 0.00(0.00) & 0.00(0.00) & 0.00(0.00)\\      \hline
  $HIM_4$ & 0.00(0.04) & 0.00(0.00) & 0.00(0.00) & 0.00(0.00)\\      \hline
  $HIM_5$ & 0.19(0.06) & 0.22(0.09) & 0.18(0.02) & 0.21(0.07)\\      \hline
  $HIM_6$ & 0.00(0.04) & 0.00(0.00) & 0.00(0.00) & 0.00(0.00)\\      \hline
\end{tabular}
\end{table}

\begin{table}[!htb]
\centering
    \captionof{table}{Accuracy metrics for S1-S3 using NOTEARS and a sample size of $500$ averaged across 100 simulated datasets.}
    \label{accs135n}
\begin{tabular}{clll}
  \hline
  \multirow{2}{*}{ } 
      & \multicolumn{3}{c}{$n=500$}\\
      \cline{2-4}
  & FDR & TPR & SHD \\  \hline
  S1 & 0.00(0.00) & 1.00(0.00) & 0.00(0.00)\\      \hline
  S2 & 0.09(0.09) & 0.90(0.11) & 3.06(3.14)\\      \hline
  S3 & 0.10(0.09) & 0.94(0.08) & 2.52(1.97)\\      \hline
\end{tabular} 
\end{table}

\begin{table}[!htb]
    \centering
    \captionof{table}{Accuracy metrics for S1-S3 using DAG-GNN and a sample size of $500$ averaged across 100 simulated datasets.}
    \label{accs135d}
\begin{tabular}{clll}
  \hline
  \multirow{2}{*}{ } 
      & \multicolumn{3}{c}{$n=500$}\\
      \cline{2-4}
  & FDR & TPR & SHD \\  \hline
   S1 & 0.03(0.09) & 0.97(0.08) & 0.14(0.42)\\      \hline
  S2 & 0.05(0.04) & 0.97(0.07) & 1.43(1.51)\\      \hline
  S3 & 0.03(0.05) & 0.98(0.04) & 0.70(1.05)\\      \hline
\end{tabular} 
\end{table}

\begin{figure}[!htb]
\centering
\begin{subfigure}{.2\linewidth}
  \centering
  \includegraphics[width=.9\linewidth]{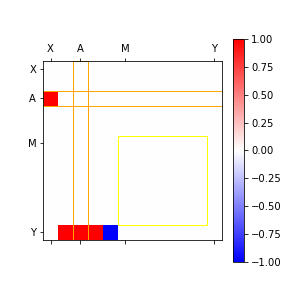}
  \caption{True}
\end{subfigure}%
\begin{subfigure}{.2\linewidth}
  \centering
  \includegraphics[width=.9\linewidth]{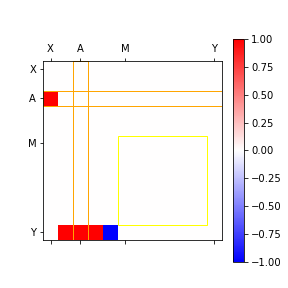}
   \caption{NOTEARS($n=500$)}
\end{subfigure}%
\begin{subfigure}{.2\linewidth}
  \centering
  \includegraphics[width=.9\linewidth]{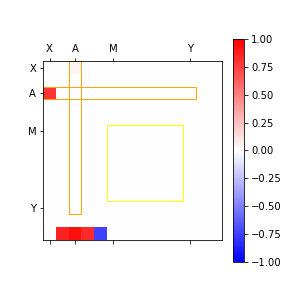}
   \caption{DAG-GNN($n=500$)}
\end{subfigure}%
\begin{subfigure}{.2\linewidth}
  \centering
  \includegraphics[width=.9\linewidth]{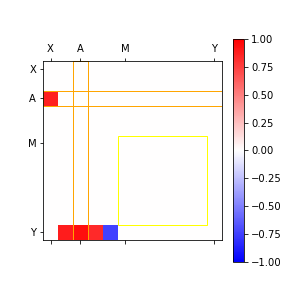}
  \caption{ISL($n=500$)}
\end{subfigure}%
\begin{subfigure}{.2\linewidth}
  \centering
  \includegraphics[width=.9\linewidth]{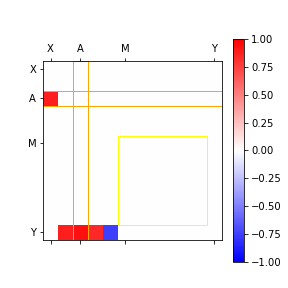}
  \caption{ISL($n=1000$)}
\end{subfigure}%
\caption{Estimated Causal Graph for Scenario 1 without mediator effects for $n=500$ and $1000$ with NOTEARS and DAG-GNN shown for comparison.}
    \label{s1b}
\end{figure}

\begin{figure}[!htb]
\centering
\begin{subfigure}{.2\linewidth}
  \centering
  \includegraphics[width=.9\linewidth]{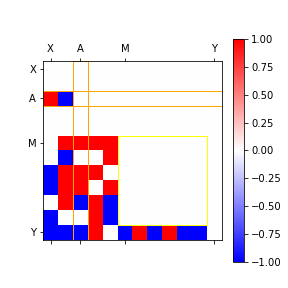}
  \caption{True}
\end{subfigure}%
\begin{subfigure}{.2\linewidth}
  \centering
  \includegraphics[width=.9\linewidth]{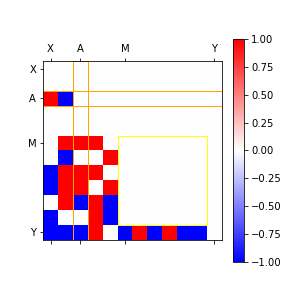}
   \caption{NOTEARS($n=500$)}
\end{subfigure}%
\begin{subfigure}{.2\linewidth}
  \centering
  \includegraphics[width=.9\linewidth]{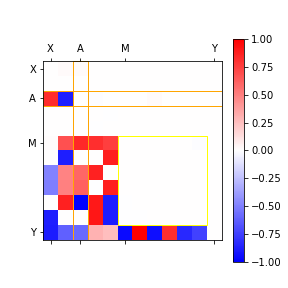}
   \caption{DAG-GNN($n=500$)}
\end{subfigure}%
\begin{subfigure}{.2\linewidth}
  \centering
  \includegraphics[width=.9\linewidth]{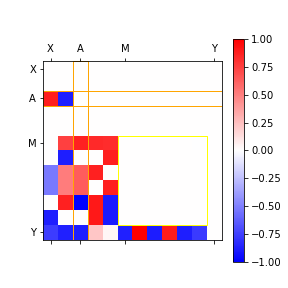}
  \caption{ISL($n=500$)}
\end{subfigure}%
\begin{subfigure}{.2\linewidth}
  \centering
  \includegraphics[width=.9\linewidth]{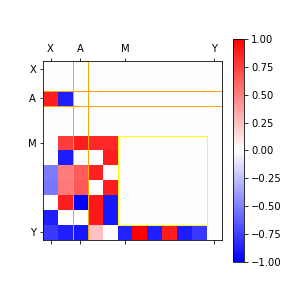}
  \caption{ISL($n=1000$)}
\end{subfigure}%
\caption{Estimated Causal Graph for Scenario 2 with parallel mediators for $n=500$ and $1000$ with NOTEARS and DAG-GNN shown for comparison.}
    \label{s2b}
\end{figure}

\begin{figure}[!htb]
\centering
\begin{subfigure}{.2\linewidth}
  \centering
  \includegraphics[width=.9\linewidth]{assets/Simulation/S3.png}
  \caption{True}
\end{subfigure}%
\begin{subfigure}{.2\linewidth}
  \centering
  \includegraphics[width=.9\linewidth]{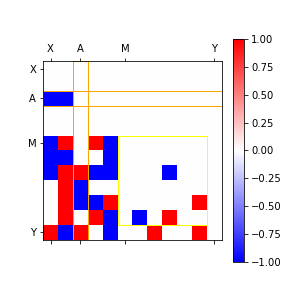}
  \caption{NOTEARS($n=500$)}
\end{subfigure}%
\begin{subfigure}{.2\linewidth}
  \centering
  \includegraphics[width=.9\linewidth]{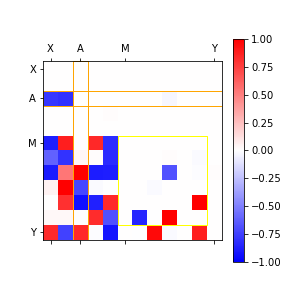}
  \caption{DAG-GNN($n=500$)}
\end{subfigure}%
\begin{subfigure}{.2\linewidth}
  \centering
  \includegraphics[width=.9\linewidth]{assets/Simulation/S3est.png}
  \caption{ISL($n=500$)}
\end{subfigure}%
\begin{subfigure}{.2\linewidth}
  \centering
  \includegraphics[width=.9\linewidth]{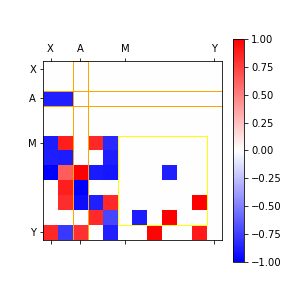}
  \caption{ISL($n=1000$)}
\end{subfigure}%
\caption{Estimated Causal Graph for Scenario 3 with sequentially ordered mediators or $n=500$ and $1000$ with NOTEARS and DAG-GNN shown for comparison.}
    \label{s3b} 
\end{figure}

\begin{figure}[!htb]
\centering
\begin{subfigure}{.3\linewidth}
  \centering
  \includegraphics[width=.9\linewidth]{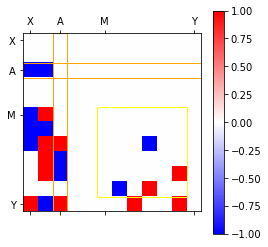}
  \caption{True}
\end{subfigure}%
\begin{subfigure}{.3\linewidth}
  \centering
  \includegraphics[width=.9\linewidth]{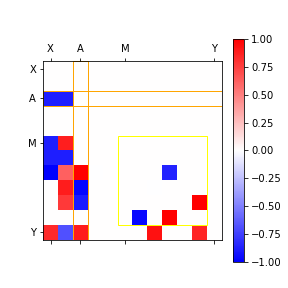}
  \caption{$n=500$}
\end{subfigure}%
\begin{subfigure}{.3\linewidth}
  \centering
  \includegraphics[width=.9\linewidth]{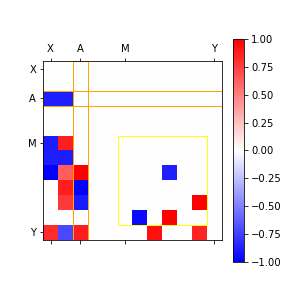}
  \caption{$n=1000$}
\end{subfigure}%
\caption{Estimated Causal Graphs for Scenario 3 when $X$ is a moderator (d to f) with sample sizes of $n=500$ and $1000$.}
    \label{s4b} 
\end{figure}

\begin{figure}[!htb]
\centering
\begin{subfigure}{.3\linewidth}
  \centering
  \includegraphics[width=.9\linewidth]{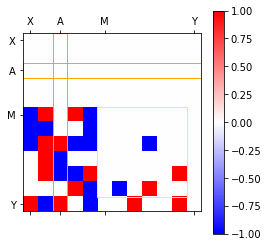}
  \caption{True}
\end{subfigure}%
\begin{subfigure}{.3\linewidth}
  \centering
  \includegraphics[width=.9\linewidth]{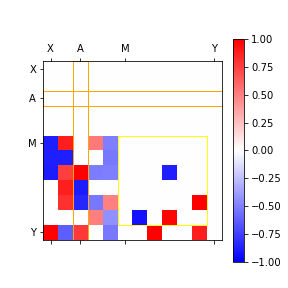}
  \caption{$n=500$}
\end{subfigure}%
\begin{subfigure}{.3\linewidth}
  \centering
  \includegraphics[width=.9\linewidth]{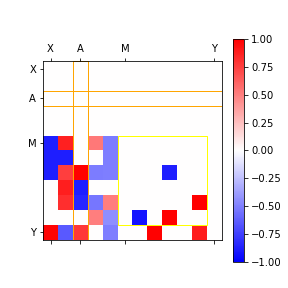}
  \caption{$n=1000$}
\end{subfigure}%
\caption{Estimated Causal Graphs for Scenario 3 with no interaction (a to c) with sample sizes of $n=500$ and $1000$.}
    \label{s5b} 
\end{figure}

\begin{figure}[!htb]
\centering
\begin{subfigure}{.3\linewidth}
  \centering
  \includegraphics[width=.9\linewidth]{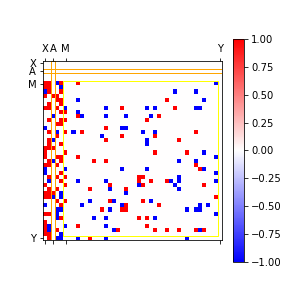}
  \caption{True S4}
\end{subfigure}%
\begin{subfigure}{.3\linewidth}
  \centering
  \includegraphics[width=.9\linewidth]{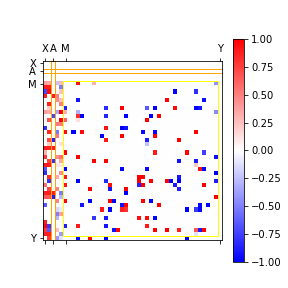}
  \caption{S4 estimate (n=1000)}
\end{subfigure}%
\begin{subfigure}{.3\linewidth}
  \centering
  \includegraphics[width=.9\linewidth]{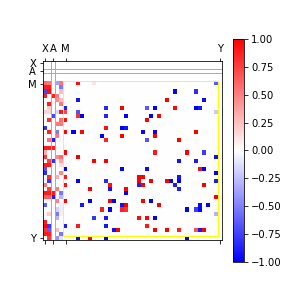}
  \caption{S4 estimate (n=1500)}
\end{subfigure}\\
\begin{subfigure}{.3\linewidth}
  \centering
  \includegraphics[width=.9\linewidth]{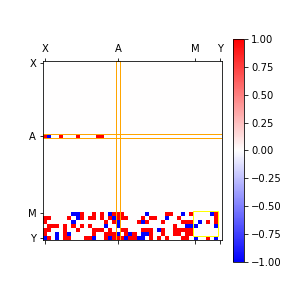}
  \caption{True S5}
\end{subfigure}%
\begin{subfigure}{.3\linewidth}
  \centering
  \includegraphics[width=.9\linewidth]{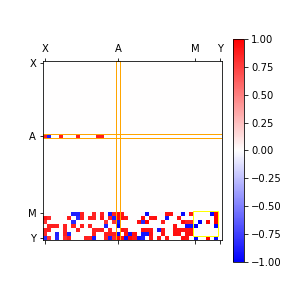}
  \caption{S5 estimate (n=1000)}
\end{subfigure}%
\begin{subfigure}{.3\linewidth}
  \centering
  \includegraphics[width=.9\linewidth]{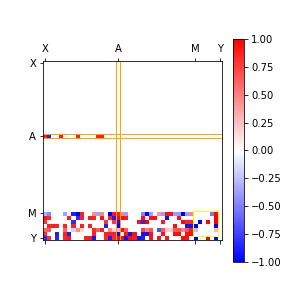}
  \caption{S5 estimate (n=1500)}
\end{subfigure}\\%
\begin{subfigure}{.3\linewidth}
  \centering
  \includegraphics[width=.9\linewidth]{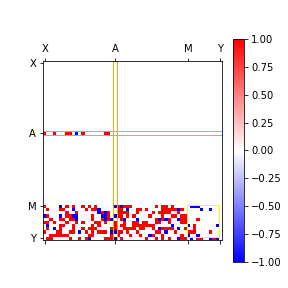}
  \caption{True S6}
\end{subfigure}%
\begin{subfigure}{.3\linewidth}
  \centering
  \includegraphics[width=.9\linewidth]{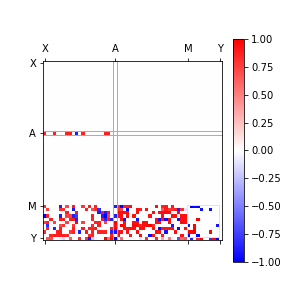}
  \caption{S6 estimate (n=1000)}
\end{subfigure}%
\begin{subfigure}{.3\linewidth}
  \centering
  \includegraphics[width=.9\linewidth]{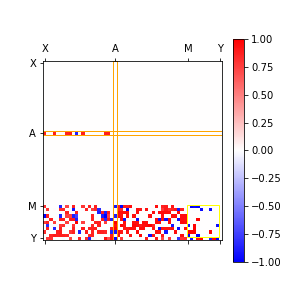}
  \caption{S6 estimate (n=1500)}
\end{subfigure}%
\caption{Results for Scenario 4 with more mediators, Scenario 5 with more moderators, and Scenario 6 with more mediators and moderators with sample sizes of $n=1000$ and $1500$.}
    \label{s6b} 
\end{figure}

\begin{table}[!htb]
    \centering
    \captionof{table}{Coverage for causal effects and interaction edge weights ($\bm{X}A\to Y$) for S3 averaged across 100 simulated datasets using 1000 bootstrap resamples each with the significance level of $0.05$.}
    \label{boot}
\begin{tabular}{cll}
\hline
  & Percentile & Gaussian\\  \hline
  $HDE$ & 0.94 & 0.88\\  \hline
  $HIE$ & 0.76 & 0.90\\      \hline
  $HDM_1$ & 1.00 & 1.00\\     \hline
  $HDM_2$ & 1.00 & 0.99\\      \hline
  $HDM_3$ & 0.96 & 0.96\\      \hline
  $HDM_4$ & 1.00 & 1.00\\      \hline
  $HDM_5$ & 1.00 & 1.00\\      \hline
  $HDM_6$ & 1.00 & 0.99\\      \hline
  $HIM_1$ & 1.00 & 1.00\\      \hline
  $HIM_2$ & 1.00 & 1.00\\      \hline
  $HIM_3$ & 1.00 & 1.00\\      \hline
  $HIM_4$ & 1.00 & 1.00\\      \hline
  $HIM_5$ & 0.97 & 0.93\\      \hline
  $HIM_6$ & 1.00 & 0.99\\      \hline
  $IX_1$ & 1.00 & 1.00\\      \hline
  $IX_2$ & 0.84 & 0.86\\      \hline
\end{tabular}
\end{table}

\begin{table}[!htb]
    \centering
    \captionof{table}{Accuracy metrics for S4-S6 with $n=1500$.}
    \label{accs46}
\begin{tabular}{clll}
  \hline
  \multirow{2}{*}{ } 
      & \multicolumn{3}{c}{$n=1500$}\\
      \cline{2-4}
  & FDR & TPR & SHD \\  \hline
  S4 & 0.01(0.01) & 0.91(0.02) & 12.28(3.07)\\      \hline
  S5 & 0.03(0.03) & 0.96(0.04) & 6.99(7.44)\\      \hline
  S6 & 0.00(0.02) & 1.00(0.02) & 1.08(6.09)\\      \hline
\end{tabular} 
\end{table}

\section{Figures and Tables from the Real Data Analysis}\label{asec:real}

\begin{table}[!ht] 
     \captionof{table}{Estimated effects for the whole population.}  \label{whole_tab} 
     \centering
\begin{tabular}{llllll}
\hline
&Insomnia&PT distress &ASD&PT PTSD &PT depression \\
\hline
 Total &0.1931& -0.0036 & 0.0572 & 0.0802 & -0.0059 \\
 \hline
Direct & 0.1160 & -0.0007 & 0.0572  & 0.0227  & -0.0021\\
\hline
Indirect & 0.0771& -0.0029 & 0 & 0.0576& -0.0038\\
\hline
\end{tabular}
\end{table}


\begin{table}[!htb]
    \centering
    \captionof{table}{Estimated causal effects for males.}
    \label{male_tab}
\begin{tabular}{llll}
\hline
            & Total   & Direct  & Indirect \\ \hline
Insomnia    & 0.2418   & 0.1729   & 0.0690\\
PT distress     & -0.0023    & -0.0004    & -0.0019\\
ASD  & 0.0498  & 0.0498  & 0\\
PT PTSD & 0.0831  & 0.0235 & 0.0597\\
PT Depression & -0.0108 & -0.0039 & -0.0069\\
\hline
\end{tabular}
\end{table}

\begin{table}[!htb]
    \centering
    \captionof{table}{Estimated causal effects for females.}
    \label{female_tab}
\begin{tabular}{llll}
\hline
            & Total   & Direct  & Indirect \\ \hline
Insomnia    & 0.102 & 0.0099 & 0.0921\\
PT distress     & -0.006    & -0.0011    & -0.0048\\
ASD  & 0.0709  & 0.0709  & 0\\
PT PTSD & 0.0748 & 0.0211 & 0.0537\\
PT Depression & 0.0034 & 0.0012 & 0.0021\\
\hline
\end{tabular}
\end{table}

 \begin{figure}[!htb]
   \centering
   \includegraphics[width=1\linewidth]{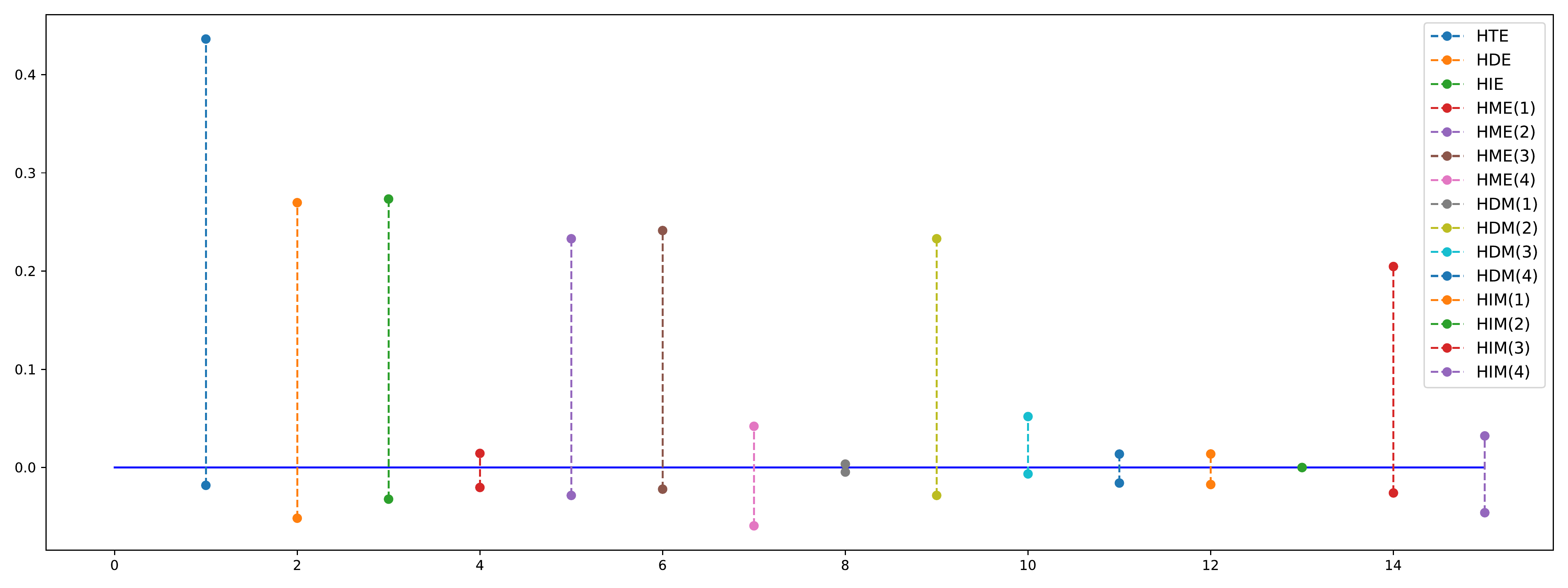}
   \captionof{figure}{$95\%$ Confidence intervals for causal effects on the population level graph, produced using 1000 bootstrap samples.}
     \label{whole_confint}
\end{figure}

 \begin{figure}[!htb]
   \centering
   \includegraphics[width=0.55\linewidth]{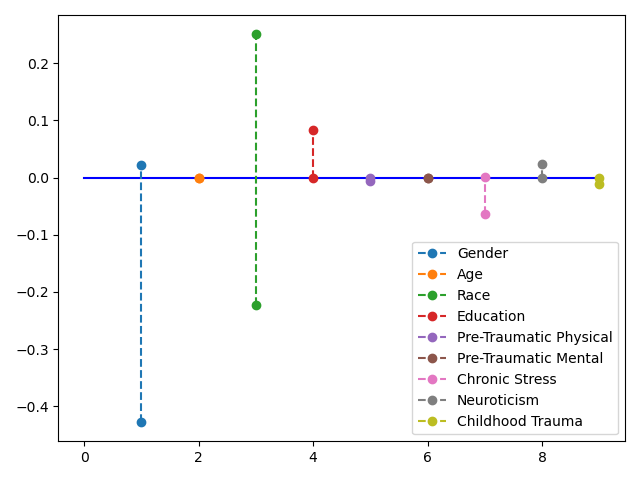}
   \caption{$95\%$ Confidence intervals for interaction effects on response computed over $1000$ bootstrap resamples.}
     \label{resp_confint}
\end{figure}
 
\begin{figure}[!htb]
   \centering
   \includegraphics[width=0.55\linewidth]{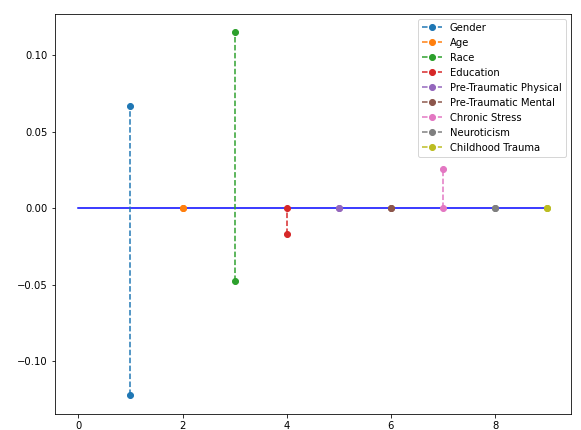}
   \caption{$95\%$ Confidence intervals for interaction effects on PT Distress computed over $1000$ bootstrap resamples.}
     \label{med1_confint}
\end{figure}
 
\begin{figure}[!htb]
   \centering
   \includegraphics[width=0.55\linewidth]{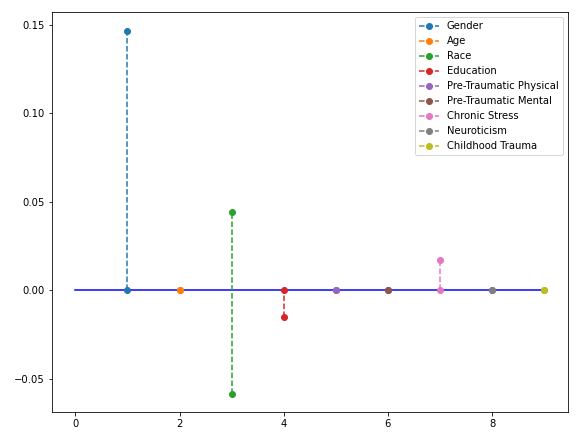}
   \caption{$95\%$ Confidence intervals for interaction effects on ASD computed over $1000$ bootstrap resamples.}
     \label{med2_confint}
\end{figure}
 
\begin{figure}[!htb]
   \centering
   \includegraphics[width=0.55\linewidth]{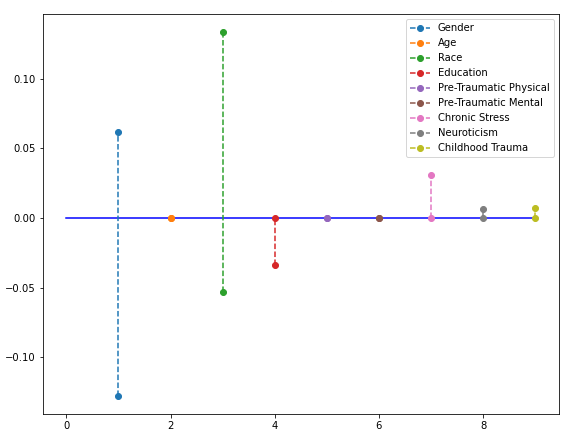}
   \caption{$95\%$ Confidence intervals for interaction effects on PT PTSD computed over $1000$ bootstrap resamples.}
     \label{med3_confint}
\end{figure}
 
\begin{figure}[!htb]
   \centering
   \includegraphics[width=0.55\linewidth]{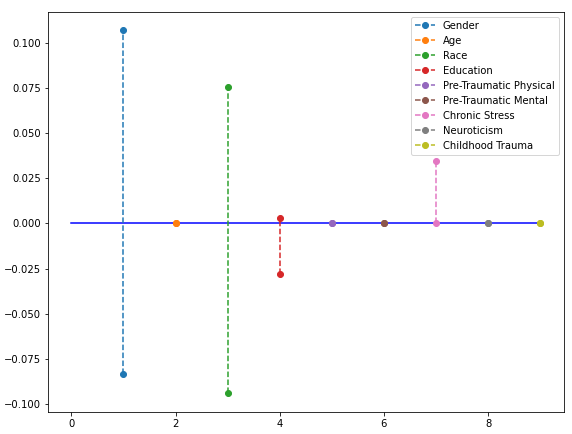}
   \caption{$95\%$ Confidence intervals for interaction effects on PT Depression computed over $1000$ bootstrap resamples.}
     \label{med4_confint}
\end{figure}

 \begin{figure}[htp]
 \centering

\subfloat[Estimated causal graph for males]{%
  \includegraphics[clip,width=.8\columnwidth]{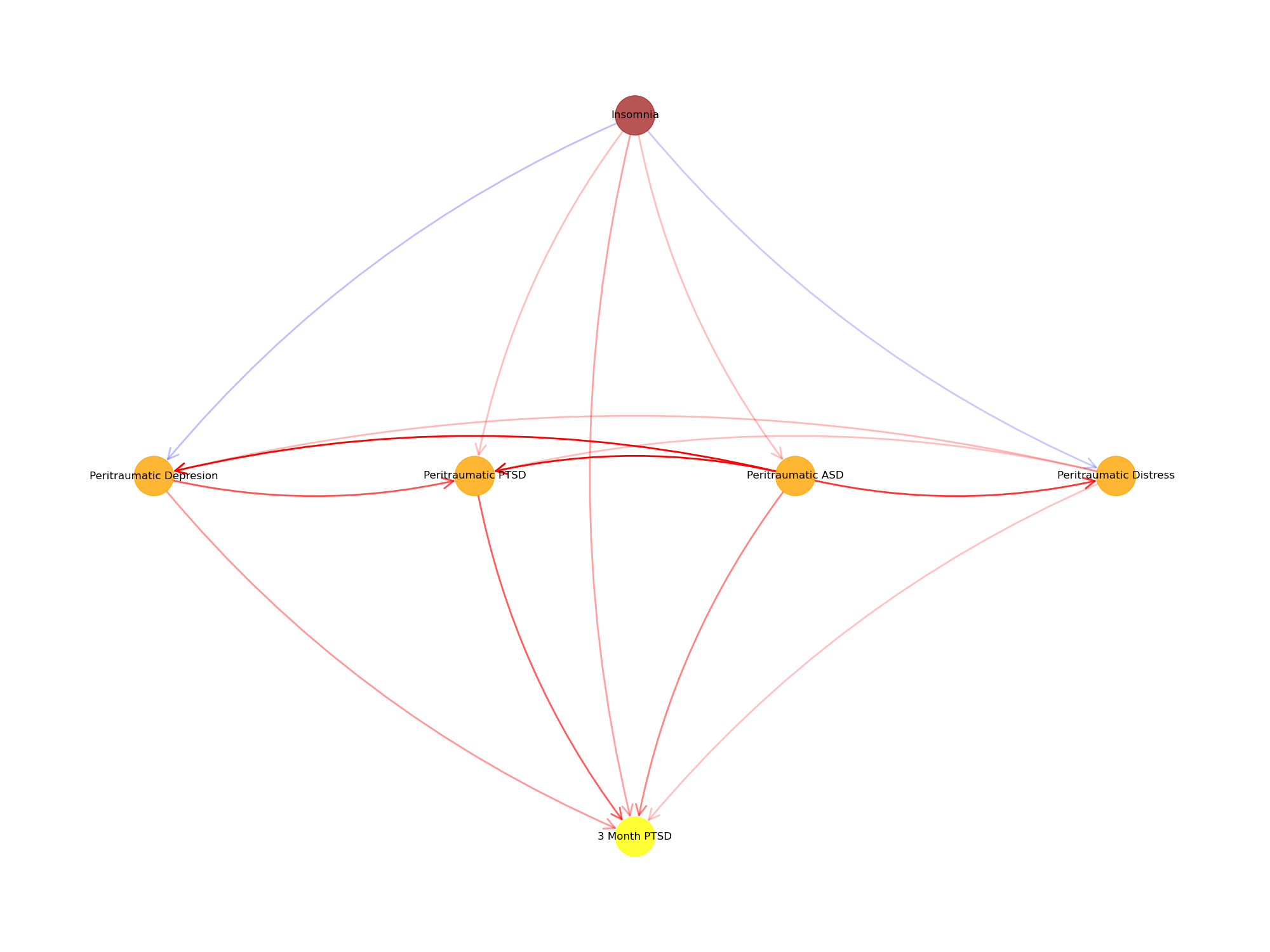}%
  \label{male_graph}
} 

\subfloat[Estimated causal graph for females]{%
  \includegraphics[clip,width=.8\columnwidth]{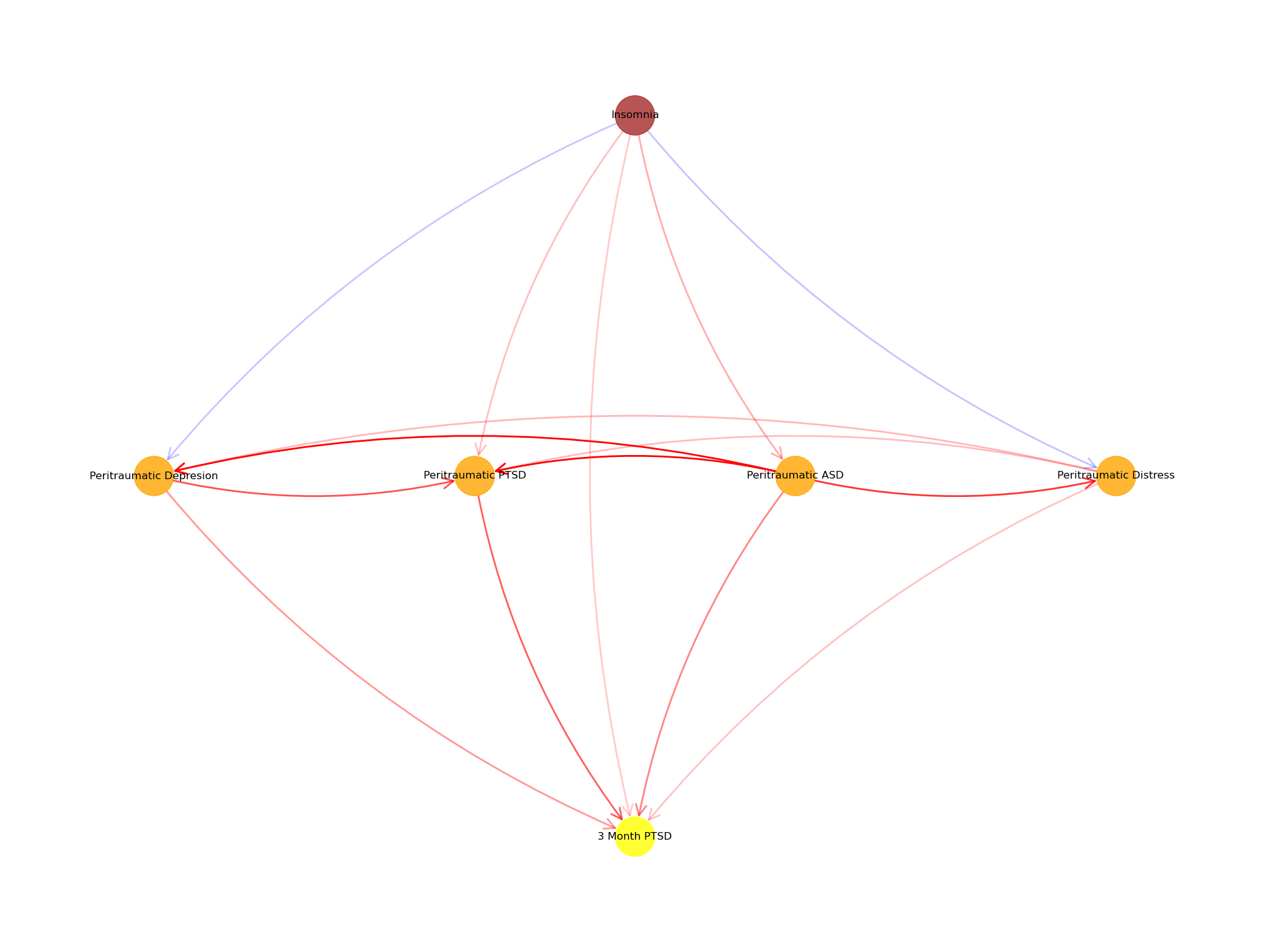}%
  \label{female_graph}
}

\caption{Causal graphs for the male and female demographic. The red node is the event of interest. Orange nodes are mediators. The yellow node is the outcome.}

\end{figure}

 \begin{figure}[htp]
 \centering

\subfloat[Estimated causal graph for the white population.]{%
  \includegraphics[clip,width=.8\columnwidth]{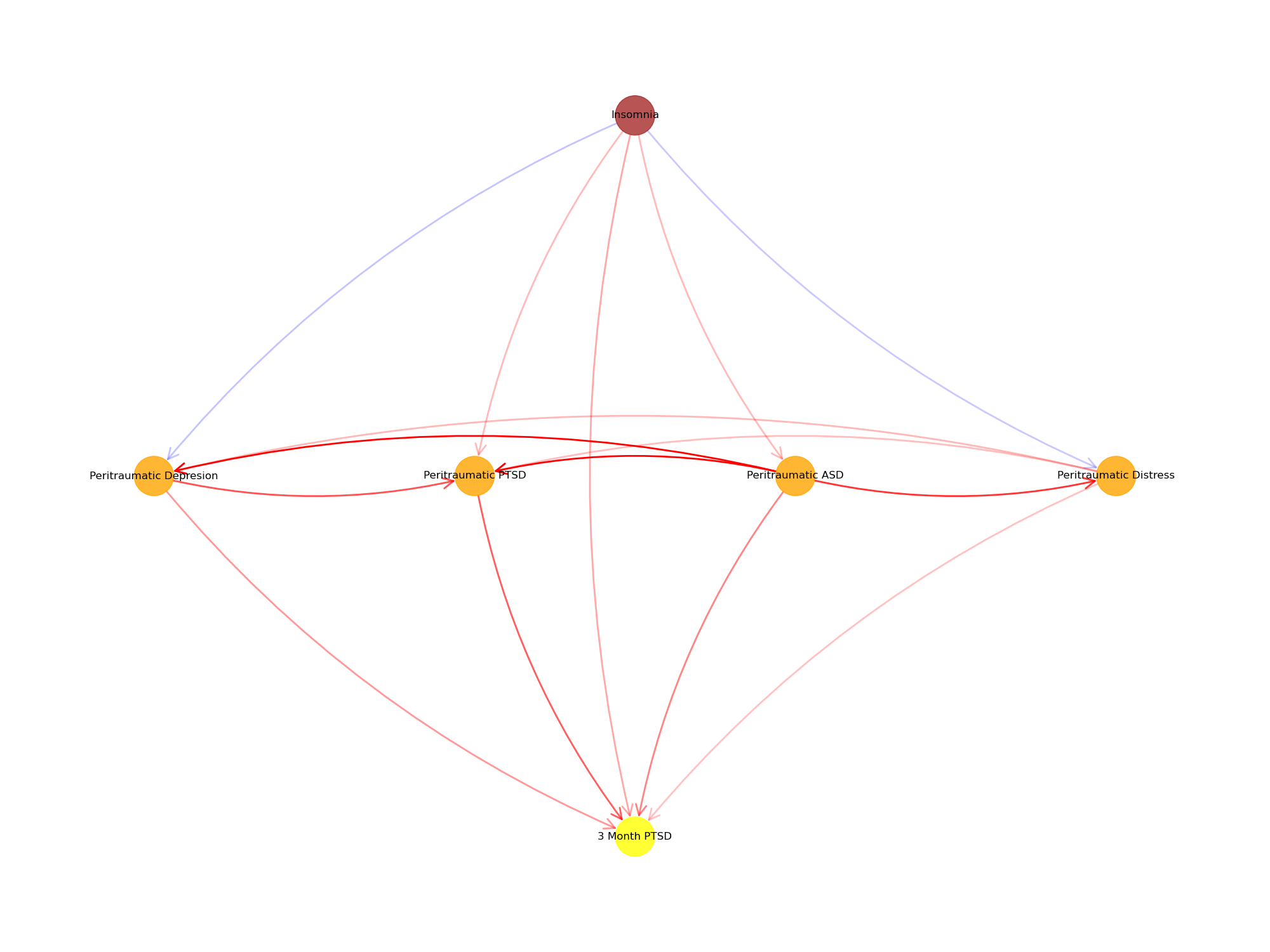}%
  \label{white_graph}
} 

\subfloat[Estimated causal graph for the black population]{%
  \includegraphics[clip,width=.8\columnwidth]{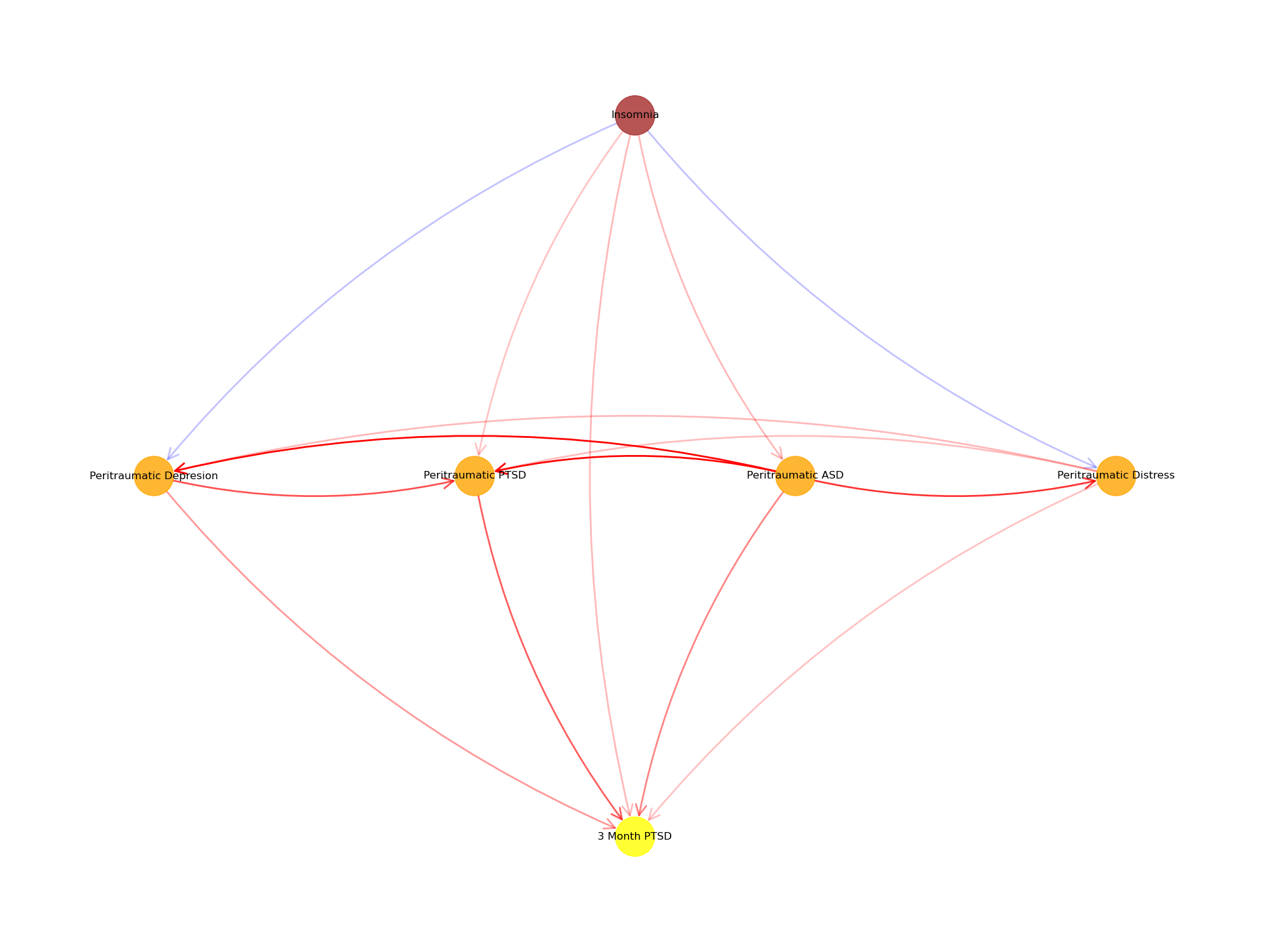}%
  \label{black_graph}
}

\caption{Causal graphs for the white and black demographic. The red node is the event of interest. Orange nodes are mediators. The yellow node is the outcome.}

\end{figure}

 \begin{figure}[htp]
 \centering

\subfloat[Estimated causal graph for the white male population.]{%
  \includegraphics[clip,width=.8\columnwidth]{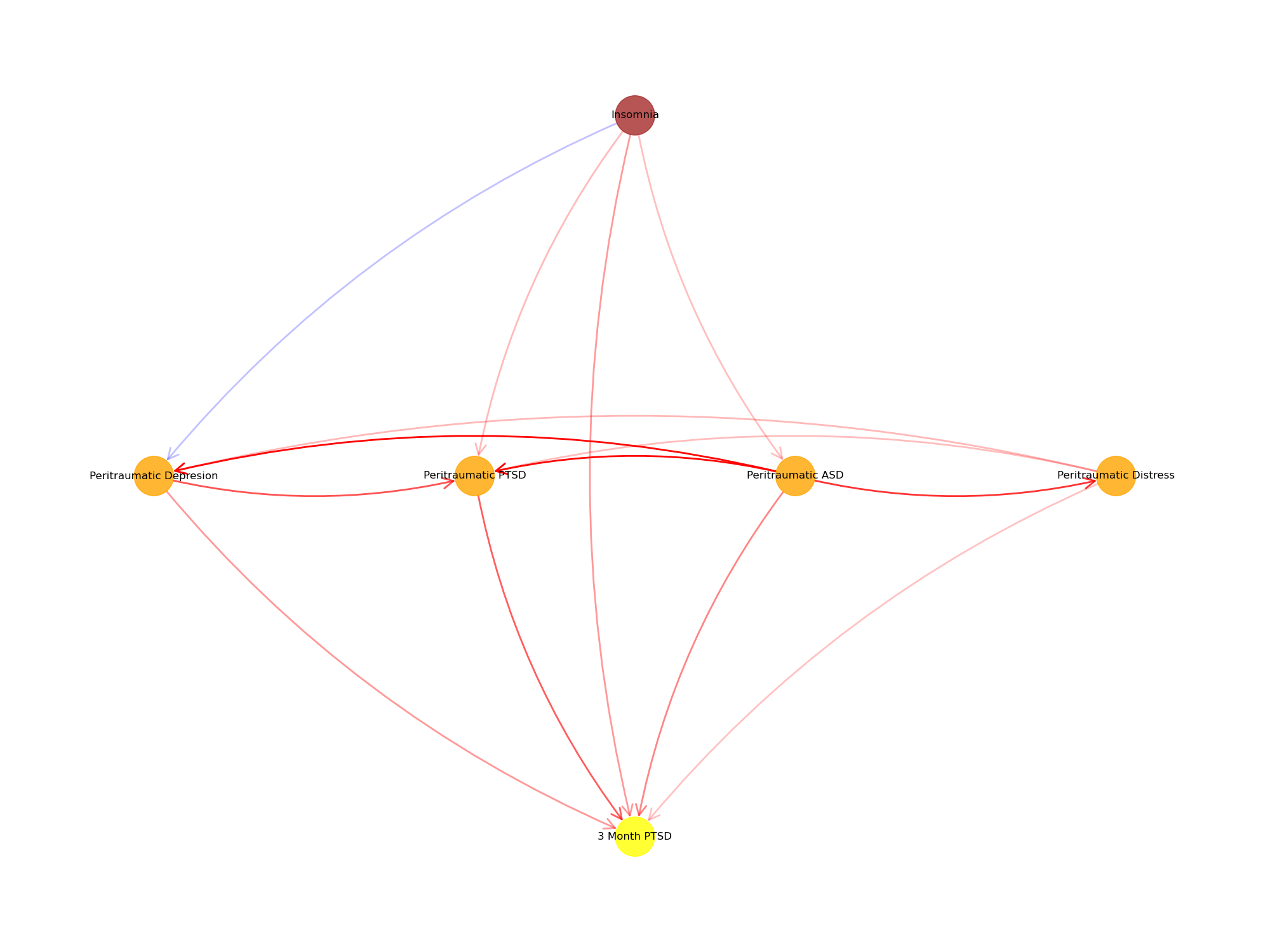}%
  \label{white_male_graph} 
}

\subfloat[Estimated causal graph for the black male population.]{%
  \includegraphics[clip,width=.8\columnwidth]{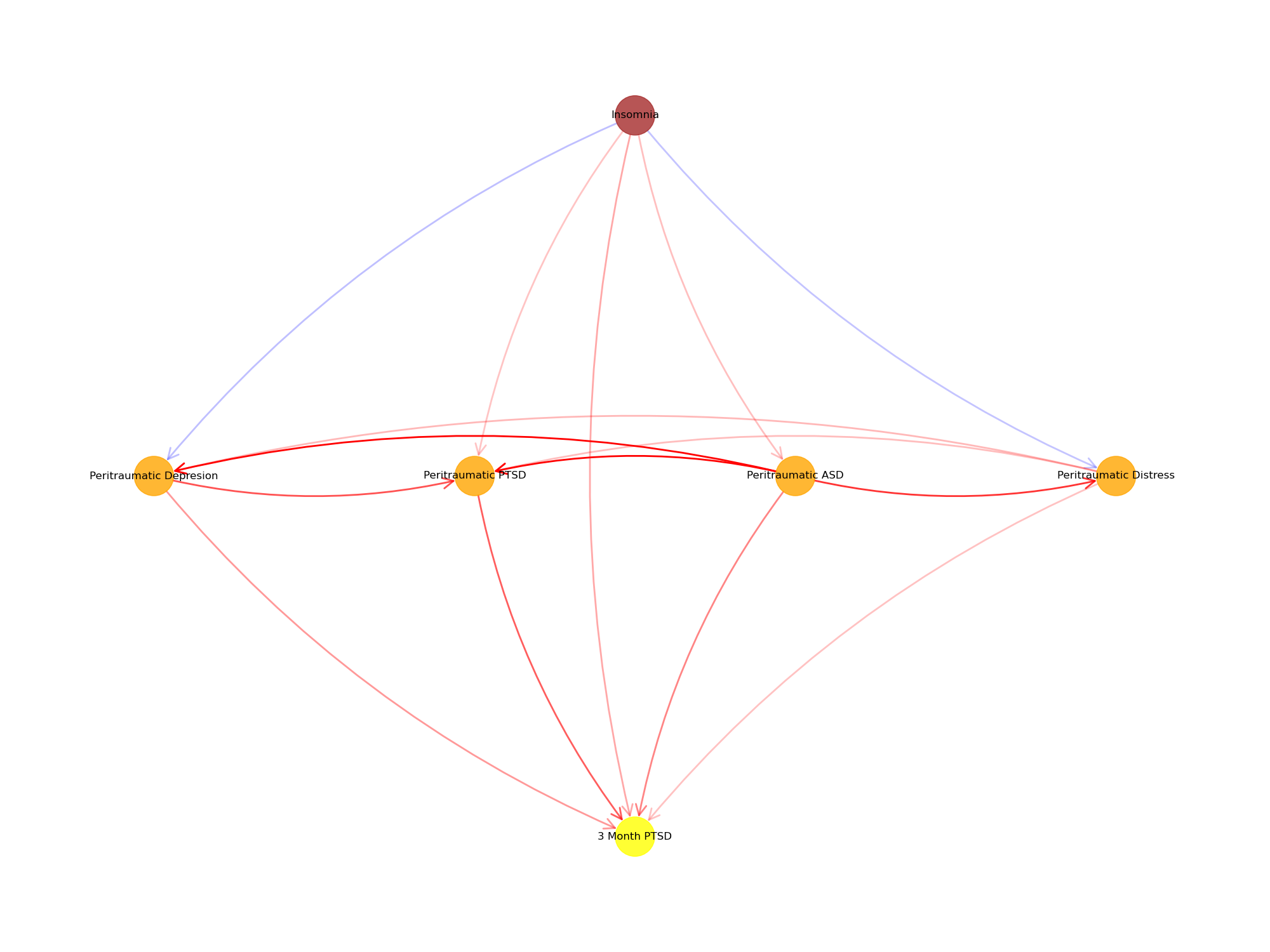}%
  \label{black_male_graph}
}

\caption{Causal graphs for the white and black male demographic. The red node is the event of interest. Orange nodes are mediators. The yellow node is the outcome.}

\end{figure}

 \begin{figure}[htp]
 \centering

\subfloat[Estimated causal graph for the white female population.]{%
  \includegraphics[clip,width=.8\columnwidth]{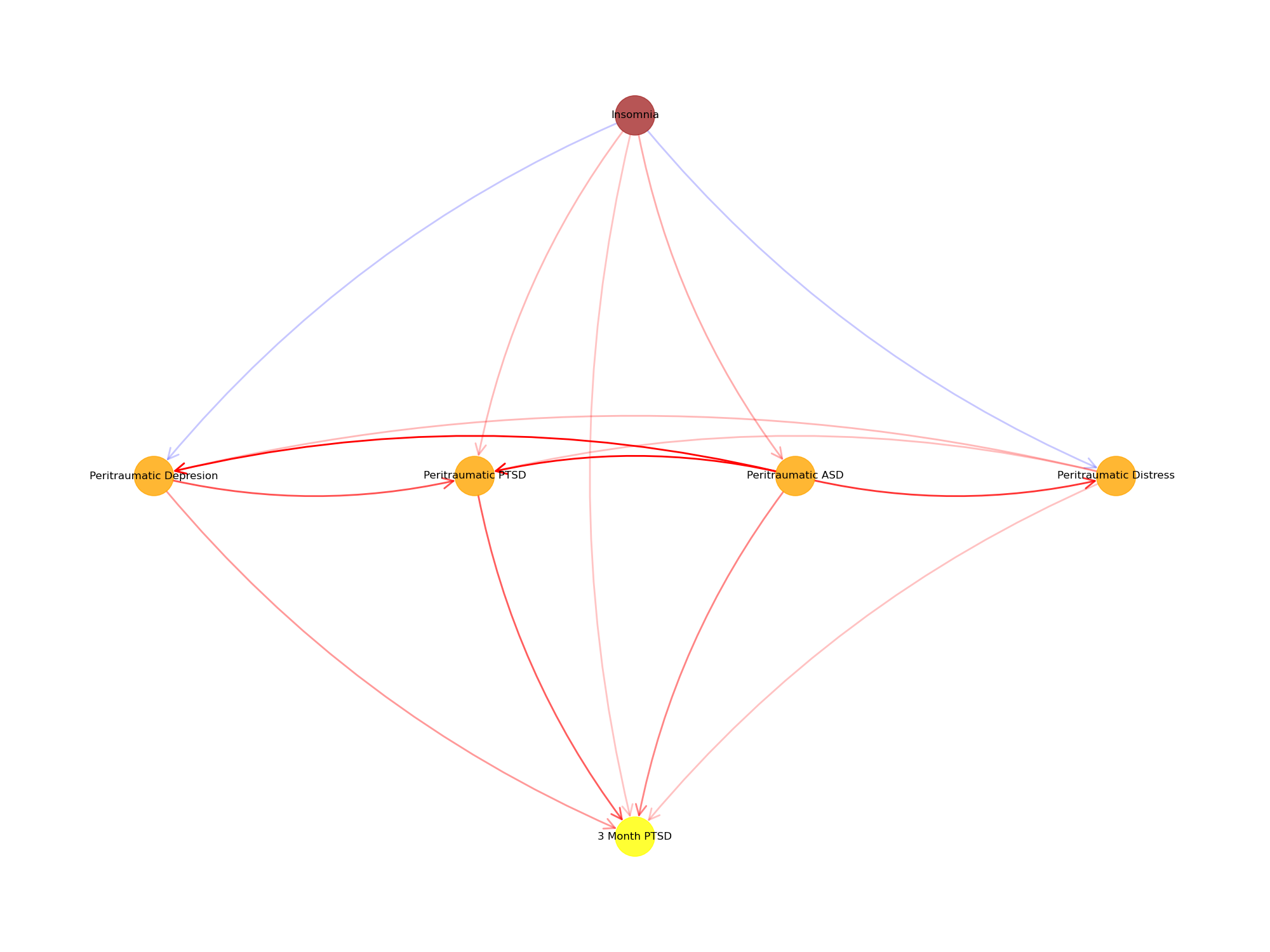}%
  \label{white_female_graph}
} 

\subfloat[Estimated causal graph for the black female population.]{%
  \includegraphics[clip,width=.8\columnwidth]{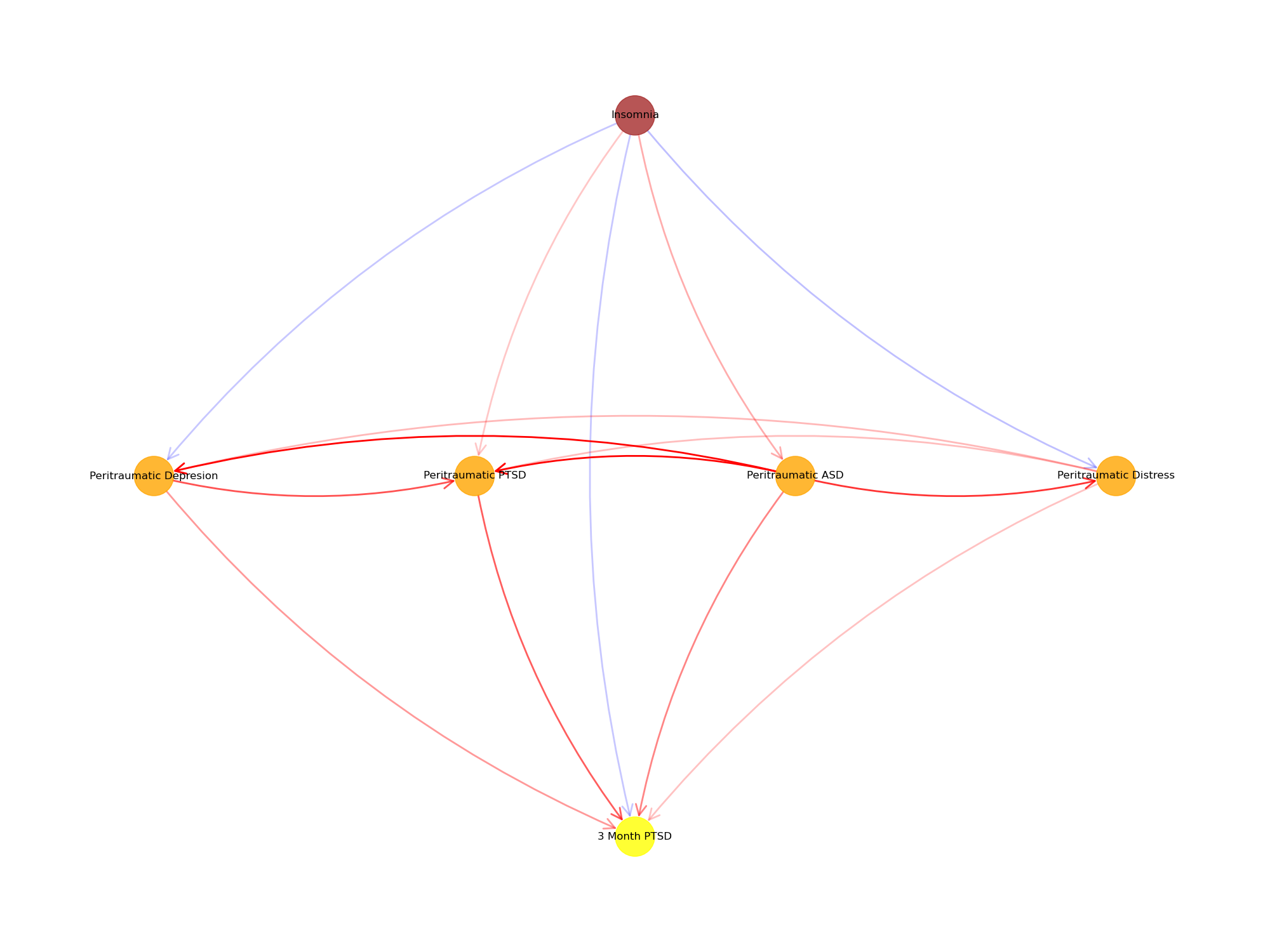}%
  \label{black_female_graph}
}

\caption{Causal graphs for the white and black female demographic. The red node is the event of interest. Orange nodes are mediators. The yellow node is the outcome.}

\end{figure}

\section{Proofs}\label{asec:proof}
\subsection{Proof of Theorem \ref{mainthm}}
\begin{theorem}
Under assumptions (A1-A3) and Model (\ref{LSEM}), we have:
\begin{itemize}
    \item $HDE(\bm{x}) = \gamma_A + \bm{\gamma}_{\bm{X}A}\bm{x}$;
    \item $HIE(\bm{x})=\bm{\gamma_M}(\bm{I}_s - \bm{B_M}^\top)^{-1}\big(\bm{\beta}_A + \bm{B}^\top_{\bm{X}A}\bm{x}\big);$
    \item $HTE(\bm{x})=HDE(\bm{x})+HIE(\bm{x}),$
\end{itemize}
where $\bm{I}_s$ is a $s\times s$ identity matrix and $\bm{x}$ is is the value of $\bm{X}$.
\end{theorem}

\begin{proof}
We begin with the systems of equations described by  Model (\ref{LSEM}) as
\begin{equation}
\label{system}
    \begin{cases}
\bm{X} &= \bm{\epsilon_X},\\
A &= \bm{\delta_X X} + \epsilon_A,\\
\bm{M} &= \bm{B^\top_X X} + \bm{\beta}_A A + \bm{B^\top_{XA} X}A + \bm{B^\top_M M} + \bm{\epsilon_M},\\
Y &= \bm{\gamma_X X} + \gamma_A A + \bm{\gamma}_{\bm{X}A}\bm{x}A + \bm{\gamma_M M} + \epsilon_Y.
\end{cases}
\end{equation}

Under the assumptions, we have the following:
\begin{align*}
    E[\bm{M}|do(A=a),\bm{X}=\bm{x}] &= E[\bm{M}|A=a,\bm{X}=\bm{x}],\\
    E[Y|do(A=a),\bm{X}=\bm{x}] &= E[Y|A=a,\bm{X}=\bm{x}],\\
    E[Y|do(A=a,\bm{M}=\bm{m}),\bm{X}=\bm{x}] &= E[Y|A=a,\bm{M}=\bm{m},\bm{X}=\bm{x}].
\end{align*}
We can then use Equation \ref{system} to compute each of these expectations:
\begin{align}
\label{A3-computed}
    E[\bm{M}|A=a,\bm{X}=\bm{x}] =\ &\bm{(I_s - B^\top_M)^{-1}} \big(\bm{B^\top_X X} + \bm{\beta}_A A + \bm{B^\top_{XA} X}A\big),
\end{align}

\begin{align}
\label{A1-computed}
    E[Y|A=a,\bm{X}=\bm{x}] =\ &E[\bm{\gamma_X X} + \gamma_A A + \bm{\gamma}_{\bm{X}A}\bm{x}A + \bm{\gamma_M M} + \epsilon_Y|A=a,\bm{X}=\bm{x}]\nonumber\\
    =\ &\bm{\gamma_X x} + \gamma_A a + \bm{\gamma}_{\bm{X}A}\bm{x}a + \bm{\gamma_M}E[\bm{M}|A=a,\bm{X}=\bm{x}]\nonumber\\
    =\ &\bm{\gamma_X x} + \gamma_A a + \bm{\gamma}_{\bm{X}A}\bm{x}a + \bm{\gamma_M}\bm{(I_s - B^\top_M)^{-1}}\big(\bm{B^\top_X x} + \bm{\beta}_A a + \bm{B^\top_{XA} x}a\big),
\end{align}
and
\begin{align}
\label{A2-computed}
    E[Y|&A=a,\bm{M}=\bm{m}^{(a)},\bm{X}=\bm{x}] = \bm{\gamma_X x} + \gamma_A a + \bm{\gamma}_{\bm{X}A}\bm{x}a + \bm{\gamma_M m^{(a)}}.
\end{align}

From Equations \ref{A1-computed} and \ref{A2-computed}, we can compute the HDE and HIE as follows,
\begin{align*}
    HDE =\ &E[Y|do(A=a+1,\bm{M}=\bm{m}^{(a)}),\bm{X}=\bm{x}] - E[Y|do(A=a)|\bm{X}=\bm{x}]\\
    =\ &E[Y|A=a+1,\bm{M}=\bm{m}^{(a)},\bm{X}=\bm{x}] - E[Y|A=a|\bm{X}=\bm{x}]\\
    =\ &\bm{\gamma_X x} + \gamma_A (a+1) + \bm{\gamma}_{\bm{X}A}\bm{x}(a+1) + \bm{\gamma_M m^{(a)}} - (\bm{\gamma_X x} + \gamma_A a + \bm{\gamma}_{\bm{X}A}\bm{x}a + \bm{\gamma_M m^{(a)}})\\
    =\ &\gamma_A + \bm{\gamma}_{\bm{X}A}\bm{x},
\end{align*}
and
\begin{align*}
    HIE =\ &E[Y|do(A=a,\bm{M}=\bm{m}^{(a+1)}),\bm{X}=\bm{x}] - E[Y|do(A=a)|\bm{X}=\bm{x}]\\
    =\ &E[Y|A=a,\bm{M}=\bm{m}^{(a+1)},\bm{X}=\bm{x}] - E[Y|A=a|\bm{X}=\bm{x}]\\
    =\ &\bm{\gamma_X x} + \gamma_A a + \bm{\gamma}_{\bm{X}A}\bm{x}a + \bm{\gamma_M m^{(a+1)}} - (\bm{\gamma_X x} + \gamma_A a + \bm{\gamma}_{\bm{X}A}\bm{x}a + \bm{\gamma_M m^{(a)}})\\
    =\ &\bm{\gamma_M}\bigg( \bigg[\bm{(I_s - B^\top_M)^{-1}}\big(\bm{B^\top_X x} + \bm{\beta}_A (a+1) + \bm{B^\top_{XA} x}(a+1)\big)\bigg] \\
    \ &- \bigg[\bm{(I_s - B^\top_M)^{-1}}\big(\bm{B^\top_X x} + \bm{\beta}_A a + \bm{B^\top_{XA} x}a\big)\bigg]\bigg)\\
    =\ &\bm{\gamma_M}\bm{(I_s - B^\top_M)^{-1}}\big(\bm{\beta}_A + \bm{B^\top_{XA} x}\big).
\end{align*}
And finally, we have the HTE as
\[HTE = HDE + HIE = \gamma_A + \bm{\gamma}_{\bm{X}A}\bm{x} + \bm{\gamma_M}\bm{(I_s - B^\top_M)^{-1}}\big(\bm{\beta}_A + \bm{B^\top_{XA} x}\big).\]
\end{proof}

\subsection{Proof of Theorem \ref{secthm}}

\begin{theorem}
Under assumptions (A1-A3) and  Model (\ref{LSEM}), we have:

\begin{itemize}
    \item $HDM_i(\bm{x}) = \{\bm{\gamma_M}\}_i\{(\bm{I}_s - \bm{B_M}^\top)^{-1}\big(\bm{\beta}_A + \bm{B}^\top_{\bm{X}A}\bm{x}\big)\}_i;$
    \item $\sum_{i=1}^s HDM_i (\bm{x}) = HIE(\bm{x});$
    \item $HIM_i(\bm{x})=HTM_i(\bm{x}) - HDM_i(\bm{x});$
    \item $HTM_i(\bm{x})=HIE(\bm{x}) - HIE_{\mathbb{G}(-i)}(\bm{x})$
\end{itemize}

where $\{\cdot\}_i$ is the $i$th element of a vector and $HIE_{\mathbb{G}(-i)}$ is the HIE under the causal graph $\mathbb{G}(-i)$ in which the $i$th mediator is removed from the original causal graph $\mathbb{G}$.
\end{theorem}

\begin{proof}

We begin by demonstrating that the $HTM_i$ can be decomposed into $HDM_i$ and $HIM_i$ using the definitions. As the second term in the product,
\[E\{M_i|do(A=a+1),\bm{X}=\bm{x}\}-E\{M_i|do(A=a),\bm{X}=\bm{x}\},\]
is included in all mediation effects, it is only necessary to show that
\begin{align}
    \label{IMEsec}
E\{Y|do(M_i=m_i+1),\bm{X}=\bm{x}\}-E\{Y|do(M_i=m_i),\bm{X}=\bm{x}\}
\end{align}
can be decomposed into
\begin{align}
\label{IDMsec}
E[Y| do(A=a,M_i=m_i^{(a)}+1,\bm{\Omega}_i=\bm{o}_i^{(a)}),\bm{X}=\bm{x}]  - E[Y|do(A=a),\bm{X}=\bm{x}]
\end{align}
and
\begin{align}
\label{HIMsec}
E[Y|do(A=a,M_i=m_i^{(a)}+1),\bm{X}=\bm{x}] - E[Y|do(A=a,M_i=m_i^{(a)}+1,\bm{\Omega}_i=\bm{o}_i^{(a)}),\bm{X}=\bm{x}].
\end{align}

Adding \ref{IDMsec} and \ref{HIMsec}, we have
\begin{align*}
    E[Y| do(A=a,M_i=m_i^{(a)}+1),\bm{X}=\bm{x}] - E[Y|do(A=a),\bm{X}=\bm{x}]
\end{align*}
which can be shown is equivalent to \ref{IMEsec} on the interventional level in the following manner:
\begin{align*}
    \ & E\{Y|do(M_i=m_i+1),\bm{X}=\bm{x}\} -E\{Y|do(M_i=m_i),\bm{X}=\bm{x}\}\\
    =\ &E\{Y|do(M_i=m_i^{(a)}+1),\bm{X}=\bm{x}\}-E\{Y|do(M_i=m_i^{(a)}),\bm{X}=\bm{x}\}\\
    =\ &E\{Y|do(A=a,M_i=m_i^{(a)}+1),\bm{X}=\bm{x}\} - E\{Y|do(A=a,M_i=m_i^{(a)}),\bm{X}=\bm{x}\}\\
    =\ &E[Y|do(A=a,M_i=m_i^{(a)}+1),\bm{X}=\bm{x}] - E[Y|do(A=a),\bm{X}=\bm{x}],
\end{align*}
where the first equality is true because the mediator value chosen, $m_i$, is arbitrary and the effect of increasing this value by a unit is assumed to be constant. The latter equalities are true by definition. Thus we have that
\begin{equation}
    \label{imedecomp}
    HTM_i(\bm{x}) = HDM_i(\bm{x}) + HIM_i(\bm{x}).
\end{equation}
This means that if we can compute the HDM for the $i$th mediator, then we can use the corresponding HTM to calculate the HIM for that mediator. From Equation \ref{A3-computed}, we have
\begin{align*}
    &   E\{M_i|do(A=a+1),\bm{X}=\bm{x}\}-E\{M_i|do(A=a),\bm{X}=\bm{x}\}\\
     =& \{\bm{(I_s - B^\top_M)^{-1}}\big(\bm{B^\top_X x} + \bm{\beta}_A(a+1) + \bm{B^\top_{XA} x}(a+1)\big)\}_i \\
     &- \{\bm{(I_s - B^\top_M)^{-1}}\big(\bm{B^\top_X x} + \bm{\beta}_Aa + \bm{B^\top_{XA} x}a\big)\}_i\\
    =& \{\bm{(I_s - B^\top_M)^{-1}}\big(\bm{\beta}_A + \bm{B^\top_{XA} x}\big)\}_i.
\end{align*}

Likewise, from Equation \ref{A2-computed}, we have
\begin{align*}
   \ & E[Y|do(A=a,M_i=m_i^{(a)}+1,\bm{\Omega}_i=\bm{o}_i^{(a)}),\bm{X}=\bm{x}] - E[Y|do(A=a),\bm{X}=\bm{x}]\\
    =\ &\big(\bm{\gamma_X x} + \gamma_A a + \bm{\gamma}_{\bm{X}A}\bm{x}a + \bm{\gamma_M (m^{(a)} + e_i)}\big)- \big(\bm{\gamma_X x} + \gamma_A a + \bm{\gamma}_{\bm{X}A}\bm{x}a + \bm{\gamma_M m^{(a)}}\big)\\
    =\ &\bm{\gamma_M e_i} = \{\bm{\gamma_M}\}_i,
\end{align*}
where $\bm{e_i}$ is the $i$th basis vector in $\bm{R}^s$. Putting these all together, we can compute the HDM of the $i$-th mediator:
\begin{align*}
    HDM_i(\bm{x}) =\ &\bigg[E[M_i|do(A=a+1),\bm{X}=\bm{x}] - E[M_i|do(A=a),\bm{X}=\bm{x}]\bigg] \\
    &\times \bigg[ E[Y|do(A=a,M_i=m_i^{(a)}+1,\bm{\Omega}_i=\bm{o}_i^{(a)}),\bm{X}=\bm{x}] - E[Y|do(A=a),\bm{X}=\bm{x}]\bigg]\\
    =\ &\{\bm{\gamma_M}\}_i \times \{\bm{(I_s - B^\top_M)^{-1}}\big(\bm{\beta}_A + \bm{B^\top_{XA} x}\big)\}_i,
\end{align*}

It is then easy to see that
\begin{align*}
\sum_{i=1}^sHDM_i(\bm{x}) =& \sum_{i=1}^s \{\bm{\gamma_M}\}_i \times \{\bm{(I_s - B^\top_M)^{-1}}\big(\bm{\beta}_A + \bm{B^\top_{XA} x}\big)\}_i \\
=& \bm{\gamma_M}\bm{(I_s - B^\top_M)^{-1}}\big(\bm{\beta}_A + \bm{B^\top_{XA} x}\big) = HIE(\bm{x})
\end{align*}
By definition we have that the HIM for the $i$-th mediator is
\[HIM_i(\bm{x}) = HTM_i(\bm{x}) - HDM_i(\bm{x}).\]
 
Thus, all that is left is the $HTM_i$ which can be interpreted as the effect of the treatment $A$ on the outcome $Y$ that is mediated through mediator $M_i$, or inversely as the change in total treatment effect that is due to $M_i$ being removed from the causal graph. Thus it can be calculated as follows,
\begin{align*}
   & HTM_i(\bm{x}) = HTE(\bm{x}) - HTE_{(i)}(\bm{x})\\
    &= \big\{HDE(\bm{x})+HIE(\bm{x})\big\}-\big\{HDE_{\mathbb{G}(-i)}(\bm{x})+HIE_{\mathbb{G}(-i)}(\bm{x})\big\}\\
    &= \big\{HDE(\bm{x})-HDE_{\mathbb{G}(-i)}(\bm{x})\big\} +  \big\{HIE(\bm{x})-HIE_{\mathbb{G}(-i)}(\bm{x})\big\}\\
    &= \big\{(\gamma_A + \bm{\gamma}_{\bm{X}A}\bm{x})-(\gamma_A + \bm{\gamma}_{\bm{X}A}\bm{x})\big\} +  \big\{HIE(\bm{x})-HIE_{\mathbb{G}(-i)}(\bm{x})\big\}\\
    &= HIE(\bm{x})-HIE_{\mathbb{G}(-i)}(\bm{x}).
\end{align*}
\end{proof}

\subsection{Proof of Theorem \ref{asec:mainthm}}

\begin{proof}
Similar to before, We begin with the systems of equations described Model \ref{NonLin-SEM}
\begin{equation*}
    \begin{cases}
\bm{X} &= \bm{\epsilon_X},\\
A &= \bm{\delta_X f_A(X)} + \epsilon_A,\\
\bm{M} &= \bm{B^\top_X f_{Mx}(X)} + \bm{\beta}^\top_A \bm{f_{Ma}}(A) + \bm{B^\top_{XA} f_{Mxa}(X},A) + \bm{B^\top_M M} + \bm{\epsilon_M},\\
Y &= \bm{\gamma_X f_{Yx}(X)} + \gamma_A f_{Ya}(A) + \bm{\gamma}_{\bm{X}A}f_{Yxa}(\bm{X},A) + \bm{\gamma_M f_{Ym}(M)} + \epsilon_Y.
\end{cases}
\end{equation*}

As before, Under the assumptions, we have the following:
\begin{align*}
    E[\bm{M}|do(A=a),\bm{X}=\bm{x}] &= E[\bm{M}|A=a,\bm{X}=\bm{x}],\\
    E[Y|do(A=a),\bm{X}=\bm{x}] &= E[Y|A=a,\bm{X}=\bm{x}],\\
    E[Y|do(A=a,\bm{M}=\bm{m}),\bm{X}=\bm{x}] &= E[Y|A=a,\bm{M}=\bm{m},\bm{X}=\bm{x}].
\end{align*}
Applying Model \ref{NonLin-SEM}, we can compute each of these expectations:
\begin{align}
\label{A3-computed2}
    E[\bm{M}|A=a,\bm{X}=\bm{x}] =\ &\bm{(I_s - B^\top_M)^{-1}} \big(\bm{B^\top_X f_{Mx}(x)} + \bm{\beta}^\top_A \bm{f_{Ma}}(a) + \bm{B^\top_{XA} f_{Mxa}(x},a)\big),
\end{align}

\begin{align}
\label{A1-computed2}
    E[Y|A=a,\bm{X}=\bm{x}] =\ &E[\bm{\gamma_X f_{Yx}(X)} + \gamma_A f_{Ya}(A) + \bm{\gamma}_{\bm{X}A}f_{Yxa}(\bm{X},A) + \bm{\gamma_M f_{Ym}(M)} + \epsilon_Y|A=a,\bm{X}=\bm{x}]\nonumber\\
    =\ &\bm{\gamma_X f_{Yx}(x)} + \gamma_A f_{Ya}(a) + \bm{\gamma}_{\bm{X}A}f_{Yxa}(\bm{x},a) + \bm{\gamma_M}E[\bm{f_{Ym}(M)}|A=a,\bm{X}=\bm{x}],
\end{align}

\begin{align}
\label{A2-computed2}
    E[Y|&A=a,\bm{M}=\bm{m}^{(a)},\bm{X}=\bm{x}]=\bm{\gamma_X f_{Yx}(x)} + \gamma_A f_{Ya}(a) + \bm{\gamma}_{\bm{X}A}f_{Yxa}(\bm{x},a) + \bm{\gamma_M} \bm{f_{Ym}}(m^{(a)}).
\end{align}

Where $\bm{m^{(a)}}$ is the value $M$ would take given $A=a$. From Equations \ref{A1-computed2} and \ref{A2-computed2}, we can compute the HDE and HIE as follows,
\begin{align*}
    HDE(\bm{x},a) =\ &\frac{\delta E[Y|do(A=a,\bm{M}=\bm{m}^{(a)}),\bm{X}=\bm{x}]}{\delta a}\\
    =\ &\frac{\delta E[Y|A=a,\bm{M}=\bm{m}^{(a)},\bm{X}=\bm{x}]}{\delta a}\\
    =\ &\frac{\delta \bigg[\bm{\gamma_X f_{Yx}(x)} + \gamma_A f_{Ya}(a) + \bm{\gamma}_{\bm{X}A}f_{Yxa}(\bm{x},a) + \bm{\gamma_M} \bm{f_{Ym}}(m^{(a)})\bigg]}{\delta a}\\
    =\ &\gamma_A\frac{\delta f_{Ya}(a)}{\delta a} + \bm{\gamma}_{\bm{X}A}\frac{\delta f_{Yxa}(\bm{x},a)}{\delta a},
\end{align*}
and
\begin{align*}
    HIE(\bm{x},a) =\ &\frac{\delta E[Y|do(A=a,\bm{M}=\bm{m}^{(a)}),\bm{X}=\bm{x}]}{\delta \bm{m}^{(a)}}\frac{\delta\bm{m}^{(a)}}{\delta a}\\
    =\ &\frac{\delta E[Y|A=a,\bm{M}=\bm{m}^{(a)},\bm{X}=\bm{x}]}{\delta \bm{m}^{(a)}}\frac{\delta\bm{m}^{(a)}}{\delta a}\\
    =\ &\frac{\delta \bigg[ \bm{\gamma_X f_{Yx}(x)} + \gamma_A f_{Ya}(a) + \bm{\gamma}_{\bm{X}A}f_{Yxa}(\bm{x},a) + \bm{\gamma_M} \bm{f_{Ym}}(m^{(a)}) \bigg]}{\delta \bm{m}^{(a)}}\frac{\delta\bm{m}^{(a)}}{\delta a}\\
    =\ &\bm{\gamma_M}\frac{\delta \bm{f_{Ym}}(m^{(a)})}{\delta \bm{m}^{(a)}}\frac{\delta\bm{m}^{(a)}}{\delta a}.
\end{align*}
where
\[\frac{\delta \bm{m}^{(a)}}{\delta a} = \frac{\delta E[\bm{M}|A=a,\bm{X}=\bm{x}]}{\delta a} = \bm{(I_s - B^\top_M)^{-1}} \frac{\delta[\bm{\beta}^\top_A \bm{f_{Ma}}(a) + \bm{B^\top_{XA} f_{Mxa}(x},a)]}{\delta a}.\]
\end{proof}

\subsection{Proof of Theorem \ref{asec:secthm}}

\begin{proof}
As it was shown in the proof of theorem \ref{secthm} that the computation of the HIM and HTM is model-independent and does not have an explicit form, we only give the computation of the HDM here. From Equation \ref{A3-computed2} and \ref{A2-computed2} and Model \ref{asec:mainthm}, we have the following:
\begin{align*}
    HDM_i(\bm{x}) =\ &\frac{\delta E[Y|do(A=a,M_i=m_i^{(a)},\bm{\Omega}_i=\bm{o}_i^{(a)}),\bm{X}=\bm{x}]}{\delta m_i^{(a)}} \times \frac{\delta E[M_i|do(A=a),\bm{X}=\bm{x}]}{\delta a} \\
     =\ &\frac{\delta E[Y|A=a,M_i=m_i^{(a)},\bm{\Omega}_i=\bm{o}_i^{(a)},\bm{X}=\bm{x}]}{\delta m_i^{(a)}} \times \frac{\delta E[M_i|A=a,\bm{X}=\bm{x}]}{\delta a} \\
     =\ &\frac{\delta \bigg[ \bm{\gamma_X f_{Yx}(x)} + \gamma_A f_{Ya}(a) + \bm{\gamma}_{\bm{X}A}f_{Yxa}(\bm{x},a) + \bm{\gamma_M f_{Ym}(m^{(a)})} \bigg]}{\delta m_i^{(a)}}\frac{\delta \bm{m}^{(a)}}{\delta a}\\
     =\ &\{\bm{\gamma_M}\}_i \frac{\delta E[\bm{f_{Ym}(m^{(a)})}]}{\delta m_i^{(a)}}\frac{\delta \bm{m}^{(a)}}{\delta a} .
\end{align*}
\end{proof}

\subsection{Proof of Theorem \ref{app_thm_XM}}

\begin{proof}
We begin with the systems of equations described in Section \ref{XMext}:
\begin{equation}
\label{LSEMXM}
    \begin{cases}
\bm{X} &= \bm{\epsilon_X},\\
A &= \bm{\delta_X X} + \epsilon_A,\\
\bm{M} &= \bm{B^\top_X X} + \bm{\beta}_A A + \bm{B^\top_{XA} X}A + \bm{B^\top_M M} + \bm{\epsilon_M},\\
Y &= \bm{\gamma_X X} + \gamma_A A + \bm{\gamma}_{\bm{X}A}\bm{X}A + \bm{\gamma_M M} + \bm{X}^\top\bm{\Gamma_{XM}}\bm{M} + \epsilon_Y.
\end{cases}
\end{equation}

Noting that $\bm{X}^\top\bm{\Gamma_{XM}}\bm{M} = vec(\bm{\Gamma_{XM}})^\top(\bm{X\otimes M})$. Under the assumptions, we have the following:
\begin{align*}
    E[\bm{M}|do(A=a),\bm{X}=\bm{x}] &= E[\bm{M}|A=a,\bm{X}=\bm{x}],\\
    E[Y|do(A=a),\bm{X}=\bm{x}] &= E[Y|A=a,\bm{X}=\bm{x}],\\
    E[Y|do(A=a,\bm{M}=\bm{m}),\bm{X}=\bm{x}] &= E[Y|A=a,\bm{M}=\bm{m},\bm{X}=\bm{x}].
\end{align*}
We can then use Equation \ref{LSEMXM} to compute each of these expectations:
\begin{align}
\label{A3-computedXM}
    E[\bm{M}|A=a,\bm{X}=\bm{x}] =\ &\bm{(I_s - B^\top_M)^{-1}} \big(\bm{B^\top_X x} + \bm{\beta}_A a + \bm{B^\top_{XA} x}a\big),
\end{align}

\begin{align}
\label{A1-computedXM}
    E[Y|A=a,\bm{X}=\bm{x}] =\ &E[\bm{\gamma_X X} + \gamma_A A + \bm{\gamma}_{\bm{X}A}\bm{x}A + \bm{\gamma_M M} + \bm{X}^\top\bm{\Gamma_{XM}}\bm{M} + \epsilon_Y|A=a,\bm{X}=\bm{x}]\nonumber\\
    =\ &\bm{\gamma_X x} + \gamma_A a + \bm{\gamma}_{\bm{X}A}\bm{x}a + (\bm{\gamma_M} + \bm{x}^\top\bm{\Gamma_{XM}})E[\bm{M}|A=a,\bm{X}=\bm{x}]\nonumber\\
    =\ &\bm{\gamma_X x} + \gamma_A a + \bm{\gamma}_{\bm{X}A}\bm{x}a + (\bm{\gamma_M} + \bm{x}^\top\bm{\Gamma_{XM}})\bm{(I_s - B^\top_M)^{-1}}\big(\bm{B^\top_X x} + \bm{\beta}_A a + \bm{B^\top_{XA} x}a\big),
\end{align}
and further
\begin{align}
\label{A2-computedXM}
    E[Y|&A=a,\bm{M}=\bm{m}^{(a)},\bm{X}=\bm{x}] = \bm{\gamma_X x} + \gamma_A a + \bm{\gamma}_{\bm{X}A}\bm{x}a + (\bm{\gamma_M + \bm{x}^\top\bm{\Gamma_{XM}}) m^{(a)}}.
\end{align}

Where $\bm{m^{(a)}}$ is the value $M$ would take given $A=a$. From Equations \ref{A1-computedXM} and \ref{A2-computedXM}, we can compute the HDE and HIE as follows,
\begin{align*}
    HDE =\ &E[Y|do(A=a+1,\bm{M}=\bm{m}^{(a)}),\bm{X}=\bm{x}] - E[Y|do(A=a)|\bm{X}=\bm{x}]\\
    =\ &E[Y|A=a+1,\bm{M}=\bm{m}^{(a)},\bm{X}=\bm{x}] - E[Y|A=a|\bm{X}=\bm{x}]\\
    =\ &\bm{\gamma_X x} + \gamma_A (a+1) + \bm{\gamma}_{\bm{X}A}\bm{x}(a+1) + (\bm{\gamma_M + \bm{x}^\top\bm{\Gamma_{XM}}) m^{(a)}}\\
    \ &- (\bm{\gamma_X x} + \gamma_A a + \bm{\gamma}_{\bm{X}A}\bm{x}a + (\bm{\gamma_M} + \bm{x}^\top\bm{\Gamma_{XM}}) \bm{m^{(a)}})\\
    =\ &\gamma_A + \bm{\gamma}_{\bm{X}A}\bm{x},
\end{align*}
and
\begin{align*}
    HIE =\ &E[Y|do(A=a,\bm{M}=\bm{m}^{(a+1)}),\bm{X}=\bm{x}] - E[Y|do(A=a)|\bm{X}=\bm{x}]\\
    =\ &E[Y|A=a,\bm{M}=\bm{m}^{(a+1)},\bm{X}=\bm{x}] - E[Y|A=a|\bm{X}=\bm{x}]\\
    =\ &\bm{\gamma_X x} + \gamma_A a + \bm{\gamma}_{\bm{X}A}\bm{x}a + (\bm{\gamma_M} + \bm{x}^\top\bm{\Gamma_{XM}})\bm{m^{(a+1)}} - (\bm{\gamma_X x} + \gamma_A a + \bm{\gamma}_{\bm{X}A}\bm{x}a + (\bm{\gamma_M} + \bm{x}^\top\bm{\Gamma_{XM}})\bm{m^{(a)}})\\
    =\ &(\bm{\gamma_M} + \bm{x}^\top\bm{\Gamma_{XM}})\bigg( \bigg[\bm{(I_s - B^\top_M)^{-1}}\big(\bm{B^\top_X x} + \bm{\beta}_A (a+1) + \bm{B^\top_{XA} x}(a+1)\big)\bigg] \\
    \ &- \bigg[\bm{(I_s - B^\top_M)^{-1}}\big(\bm{B^\top_X x} + \bm{\beta}_A a + \bm{B^\top_{XA} x}a\big)\bigg]\bigg)\\
    =\ &(\bm{\gamma_M} + \bm{x}^\top\bm{\Gamma_{XM}})\bm{(I_s - B^\top_M)^{-1}}\big(\bm{\beta}_A + \bm{B^\top_{XA} x}\big).
\end{align*}
And finally, we have the HTE as
\[HTE = HDE + HIE = \gamma_A + \bm{\gamma_{XA} x} + (\bm{\gamma_M} + \bm{x}^\top\bm{\Gamma_{XM}})\bm{(I_s - B^\top_M)^{-1}}\big(\bm{\beta_A} + \bm{B^\top_{XA} x}\big).\]
\end{proof}

\end{document}

%% file: math_commands.tex
\usepackage{amsmath,amsfonts,bm,amssymb,mathtools,amsthm}

\def\eqref#1{equation~\ref{#1}}

\def\1{\bm{1}}

\DeclareMathAlphabet{\mathsfit}{\encodingdefault}{\sfdefault}{m}{sl}
\SetMathAlphabet{\mathsfit}{bold}{\encodingdefault}{\sfdefault}{bx}{n}

\newcommand\independent{\protect\mathpalette{\protect\independenT}{\perp}}
\def\independenT#1#2{\mathrel{\rlap{$#1#2$}\mkern2mu{#1#2}}}
